\documentclass[runningheads]{llncs2e/llncs}

\usepackage[utf8]{inputenc}

\usepackage{etoolbox}

\usepackage{amsmath}
\usepackage{amssymb}
\usepackage{thm-restate}
\usepackage{cancel}
\usepackage{wrapfig}

\usepackage{caption}

\usepackage{multirow}
\usepackage{bigdelim}
\usepackage{tabularx}

\usepackage{hyperref}
\usepackage{cleveref}

\usepackage{listings}
\usepackage{xcolor}

\usepackage[page]{appendix}

\usepackage{graphicx}

\usepackage{tikz}
\usetikzlibrary{arrows}
\usetikzlibrary{automata}
\usetikzlibrary{positioning}
\usetikzlibrary{shapes}
\usetikzlibrary{fit}

\usepackage{xparse}

\apptocmd{\thebibliography}{\raggedright}{}{}

\definecolor{mycomment}{rgb}{0.05,0.35,0.05}
\definecolor{mykeyword}{HTML}{4c1b54}
\definecolor{lightgreen}{HTML}{d9ead3}
\definecolor{lightblue}{HTML}{cfe2f3}
\definecolor{lightpink}{HTML}{ead1dc}
\definecolor{lightyellow}{HTML}{fff2cc}

\newcommand{\define}{\triangleq}

\newcommand{\SL}{{\textsc{SL}}}

\newcommand{\Nat}{\textnormal{\texttt{Nat}}}
\newcommand{\A}{\mathcal{A}}
\newcommand{\domain}{\mathcal{D}}
\newcommand{\interp}{\mathcal{I}}
\newcommand{\structure}[1]{\text{STRUCT}\brackets{#1}}

\newcommand{\sat}{\vDash}
\newcommand{\nsat}{\nvDash}
\newcommand{\To}{\Rightarrow}

\newcommand{\bigland}{\bigwedge}

\newcommand{\argmin}{\mathop{\arg\min}}

\newcommand{\dosub}[2]{#2 / #1}

\NewDocumentCommand{\sub}{O{x}O{\alpha}}{
	\brackets{ \dosub{#1}{#2} }%
}

\NewDocumentCommand{\subs}{O{}O{r}O{x}O{\alpha}O{}}{%
	\brackets{ #1 \dosub{#3_1}{#4_1}, \dots, \dosub{#3_{#2}}{#4_{#2}} #5 }
}

\NewDocumentCommand{\subsi}{O{}O{r}O{x}O{\alpha}O{}}{%
	\brackets{ #1 \dosub{#3_1}{#4_{i_1}}, \dots, \dosub{#3_{#2}}{#4_{i_{#2}}} #5 }
}

\newcommand{\ifelsedef}[3]
{
	\left\{
	\begin{array}{ll}
		#1 & \mbox{if } #2 \\
		#3 & \mbox{otherwise}
	\end{array}
	\right.
}

\newcommand{\sort}[1]{\boldsymbol{#1}}

\Crefname{formula}{Formula}{Formulas}
\Crefname{nclaim}{Claim}{Claims}
\Crefname{cong}{Congruence Condition}{Congruence Conditions}

\spnewtheorem{observation}{Observation}{\itshape}{\rmfamily}
\spnewtheorem{nclaim}{Claim}{\itshape}{\rmfamily}

\makeatletter
\NewDocumentEnvironment{longproof}{m}
{%
	\setcounter{proofpart}{0}%
	\setcounter{proofstep}{0}%
	\setcounter{proofcase}{0}%
	\noindent\underline{\textbf{\textit{Proof of} {#1}}}%
	\par\nobreak\smallskip%
	\@afterheading
}
{\begin{center}\textit{Q.E.D.} {#1}.\end{center}}
\makeatother

\makeatletter
\newcounter{proofpart}
\newcommand{\proofpart}[1]{%
	\stepcounter{proofpart}%
	\noindent\textbf{Part \theproofpart: #1}\par\nobreak\smallskip
	\setcounter{proofstep}{0}
	\setcounter{proofcase}{0}
	\@afterheading
}
\makeatother

\makeatletter
\newcounter{proofstep}
\newcommand{\proofstep}[1]{%
	\par
	\stepcounter{proofstep}%
	\noindent\underline{Step 
		\ifnum0<\value{proofpart}%
		\theproofpart.%
		\fi
		\theproofstep}: #1\par\nobreak\smallskip
	\setcounter{proofcase}{0}
}
\makeatother

\makeatletter
\newcounter{proofcase}
\newcommand{\proofcase}[1]{%
	\par
	\stepcounter{proofcase}%
	\noindent\textit{Case 
		\ifnum0<\value{proofpart}%
		\theproofpart.%
		\fi
		\ifnum0<\value{proofstep}%
		\theproofstep.%
		\fi
		\theproofcase}: #1\par\nobreak\smallskip
}
\makeatother

\newcommand{\atexttt}[1]{\textnormal{\texttt{#1}}}

\newcommand{\total}{\atexttt{total}}

\newcommand{\Coin}{\atexttt{Coin}}
\newcommand{\Addr}{\atexttt{Address}}

\newcommand{\ac}{\atexttt{active}}
\newcommand{\hc}{\atexttt{has-coin}}

\newcommand{\mF}{\mathcal{F}}

\newcommand{\bal}{\atexttt{bal}}

\newcommand{\parens}[1] { \left( #1 \right)}
\newcommand{\pair}[2]{\parens{#1,#2}}
\newcommand{\braces}[1]{ \left\{ #1 \right\} }
\newcommand{\brackets}[1]{ \left[ #1 \right] }
\newcommand{\abs}[1]{
	\left|#1\right|
}

\newcommand{\N}{\mathbb{N}}
\newcommand{\liff}{\leftrightarrow}
\newcommand{\ind}{\atexttt{idx}}
\newcommand{\inv}{\atexttt{inv}}

\newcommand{\old}{\atexttt{old-}}
\newcommand{\new}{\atexttt{new-}}

\newcommand{\world}{\atexttt{world}}

\newcommand{\oldworld}{\old\world}
\newcommand{\newworld}{\new\world}

\newcommand{\mint}{\atexttt{mint}}
\newcommand{\transfer}{\atexttt{transferFrom}}

\newcommand{\msum}{\atexttt{sum}}
\newcommand{\oldsum}{\old\msum}
\newcommand{\newsum}{\new\msum}
\newcommand{\oldbal}{\old\bal}
\newcommand{\newbal}{\new\bal}
\newcommand{\mcount}{\atexttt{count}}

\newcommand{\zbal}{\atexttt{z-bal}}
\newcommand{\zhc}{\atexttt{z-has-coin}}
\newcommand{\zact}{\atexttt{z-active}}
\newcommand{\zsum}{\atexttt{z-sum}}

\newcommand{\xhascoin}{\atexttt{x-has-coin}}
\newcommand{\xsum}{\atexttt{x-sum}}

\newcommand{\xbal}{\atexttt{x-bal}}
\newcommand{\xactive}{\atexttt{x-active}}

\newcommand{\ybal}{\atexttt{y-bal}}
\newcommand{\yhascoin}{\atexttt{y-has-coin}}
\newcommand{\yactive}{\atexttt{y-active}}
\newcommand{\ysum}{\atexttt{y-sum}}

\newcommand{\oldhc}{\old\hc}
\newcommand{\newhc}{\new\hc}
\newcommand{\oldact}{\old\ac}
\newcommand{\newact}{\new\ac}

\newcommand{\suma}{\sum_{a \in \Addr}}
\newcommand{\aco}{\ac_{\,\atexttt{1}}}
\newcommand{\act}{\ac_{\,\atexttt{2}}}
\newcommand{\hco}{\hc_{\,\atexttt{1}}}
\newcommand{\hct}{\hc_{\,\atexttt{2}}}

\newcommand{\ufenc}{\atexttt{uf}}
\newcommand{\intenc}{\atexttt{int}}
\newcommand{\natenc}{\atexttt{nat}}
\newcommand{\idenc}{\atexttt{id}}

\newcommand{\twocm}{two-counter machine}

\newcommand{\TwoCM}{Two-Counter Machine}
\newcommand{\twocms}{two-counter machines}
\newcommand{\pc}{\textsc{pc}}

\renewcommand{\orcidID}[1]{\href{http://orcid.org/#1}{\raisebox{-1.25pt}{\includegraphics{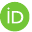}}}}

\begin{document}


\title{Summing Up Smart Transitions}

\author{%
Neta~Elad\inst{1}\orcidID{0000-0002-5503-5791}
\and
Sophie~Rain\inst{2}\orcidID{0000-0002-8940-4989}
\and
Neil~Immerman\inst{3}\orcidID{0000-0001-6609-5952}
\and
Laura Kov\'acs\inst{2}\orcidID{0000-0003-0845-5811}
\and
Mooly~Sagiv\inst{1}\orcidID{0000-0002-0723-1309}
}
\institute{%
	Tel~Aviv~University,~Israel
	\and
	TU Wien,~Austria
	\and
	UMass~Amherst,~USA
}

\maketitle


\begin{abstract}
Some of the most significant high-level properties of currencies are
the sums of certain account balances.
Properties of such sums can ensure the integrity of currencies and transactions.
For example, the sum of balances should not be changed by a transfer operation.
Currencies manipulated by code present a verification challenge to
mathematically prove their integrity by reasoning about computer
programs that
operate over them, e.g., in Solidity.
The ability to reason about sums is essential:
even the simplest ERC-20 token standard of the Ethereum community
provides a way to access the total supply of balances.

Unfortunately, reasoning about code written against this interface is non-trivial:
the number of addresses is  unbounded,
and establishing global invariants like the preservation of the sum of the balances by operations like transfer requires higher-order reasoning.
In particular, automated reasoners do not provide ways to 
specify summations of arbitrary length.

In this paper, we present a generalization of first-order logic which can express 
the unbounded sum of balances. We prove the decidablity of one of our extensions and the
undecidability of a slightly richer one.  
We introduce first-order encodings to automate 
reasoning over software transitions with summations.
We demonstrate the applicability of our results by using SMT solvers
and first-order provers for validating the correctness of common transitions in smart contracts.\\

This submission is an extended version of the CAV 2021 paper "Summing Up Smart Transitions", by N. Elad, S. Rain, N. Immerman, L. Kov\'acs and M. Sagiv. 

\end{abstract}

%
	%

\section{Introduction}

A basic challenge in smart contract verification
is how to express the functional correctness of transactions,
such as currency minting 
or transferring between accounts.
Typically, the correctness of such a transaction can be verified by
proving that the transaction leaves 
the sum of certain account balances unchanged.


Consider for example the task of minting an unbounded number of tokens
in the simplified  ERC-20 token standard of the Ethereum
community~\cite{erc-20-standard},
as illustrated in \Cref{fig:mint-n-high-level}\footnote{
	The \old{} prefix denotes the value of a function before 
	the \mint{} transition,
	and the \new{} prefix denotes the value afterwards.
}.
This example deposits the minted amount ({\tt n}) 
into the receiver's address ({\tt a})
and we need to ensure that the \mint{} operation \emph{only} changed
the balance of the receiver.
To do so, in addition to 
(i) proving that the balance of the receiver has
been increased by {\tt n}, 
we also need to verify that 
(ii) the account balance 
of every user address {\tt a}$'$
different than {\tt a} has not been changed during the \mint{} operation 
and that
(iii) the \msum{} of all balances changed exactly by
the amount that was minted. 
The validity of these three requirements (i)-(iii), 
formulated as the post-conditions of \Cref{fig:mint-n-high-level},
imply its functional correctness.

\begin{figure}[tb]
	\begin{lstlisting}[gobble=8]
		a: Address
		n: Nat
		@\listingrule@
		mint(a,n)
		@\listingrule@
		# Post-conditions
		assert new-bal(a) = old-bal(a) + n @\hfill@ #(i)
		for each Address a@$'$@ @$\neq$@ a: @\hfill@ #(ii)
			assert new-bal(a@$'$@) = old-bal(a@$'$@)
		assert new-@sum@() = old-@sum@() + n @\hfill@ #(iii)
	\end{lstlisting}
	\caption{Minting $n$ Tokens in ERC-20.}
	\label{fig:mint-n-high-level}\vspace{-0.5cm}
\end{figure}

Surprisingly, proving formulas similar to the post-conditions of \Cref{fig:mint-n-high-level} is challenging for
state-of-the-art automated reasoners, such as SMT solvers~\cite{Z3,CVC4,Yices} and
first-order provers~\cite{Vampire,E,Spass}: it
requires reasoning that links local changes of the receiver ({\tt a})
with a global state capturing the \msum{} of all balances,
as well as constructing that global state as an aggregate of an unbounded but finite
number of \Addr{} balances. 
Moreover, our encoding of the problem uses
discrete coins that are minted and deposited, whose number is unbounded but finite as well.
  

%
%
%


In this paper we address  verification challenges of  
software transactions with aggregate properties, such  
as preservation of sums by transitions that manipulate low-level, individual
entities. 
Such properties are best expressed in higher-order logic, hindering 
the use of existing automated reasoners for proving them.
To overcome such a reasoning limitation,
we introduce \emph{Sum Logic} (\SL{}) as a generalization of
first-order logic, in particular of Presburger arithmetic.
Previous works~\cite{fmt,generalized-quantifiers,counting-quantifiers}
have also introduced extensions of first-order logic with
aggregates by counting quantifiers or generalized quantifiers. 
In Sum Logic (\SL{}) we only consider the special case of integer sums over
uninterpreted functions, allowing us to formalize \SL{} properties with and
about unbounded sums, in particular sums  of account
balances, without higher-order operations (Section~\ref{sec:SL}).
We prove the decidability of one of our \SL{} extensions  and the
undecidability of a slightly richer one
(Section~\ref{sec:decidability}).
Given previous results \cite{fmt}, our
	undecidability result is not surprising. In contrast, what may be unexpected is our decidability result and the
	fact that we can use our first-order fragment for a convenient and practical new way to verify the
	correctness of smart contracts.

We further introduce first-order encodings
which enable automated reasoning over software transactions with
summations in \SL{}
(Section~\ref{sec:encoding}). Unlike~\cite{extending-smt-to-hol},
where SMT-specific extensions supporting higher-order reasoning have
been introduced,  the logical encodings we propose allow one to use existing
reasoners without any modification. We are not restricted
to SMT reasoning, but can also leverage generic automated reasoners, such
as first-order theorem provers, 
supporting first-order logic. We believe our results ease applying
automated reasoning to smart contract verification even for
non-experts.

We demonstrate the practical applicability of our results by using SMT
solvers and first-order provers
for validating the correctness of common financial transitions appearing in
\emph{smart contracts} (\Cref{sec:exp}). 
We refer to these transitions as \emph{smart transitions}. 
We encode \SL{} into pure first-order logic by adding another sort that 
represents the tokens of the crypto-currency themselves (which we dub
``coins'').

Although the encodings of \Cref{sec:encoding} do not
        translate to our   
decidable \SL{} fragment from \Cref{sec:decidability},
our experimental results show that automated reasoning
        engines 
can handle them consistently and fast.
The decidability results of \Cref{sec:encoding} set
the boundaries for what one can expect to achieve,
while our experiments from  \Cref{sec:encoding}
        demonstrate that the unknown middle-ground
can still be automated.

While our work is mainly motivated by smart contract verification,
our results can be used for arbitrary software
transactions implementing sum/aggregate properties. 
Further, when compared to the smart contract verification framework
of~\cite{azureblockchain},
we note that we are not restricted to proving the correctness of
smart contracts as finite-state machines, but can deal with semantic
properties expressing financial transactions in smart contracts, such
as currency minting/transfers. 

While ghost variable approaches \cite{SRI} can reason about changes to the global state
(the sum),
our approach allows the verifier to specify only the local changes
and automatically prove the impact on the global state.

\paragraph{Contributions.}
In summary, this paper makes the following contributions:
\begin{itemize}
	\item 
		We present a generalization to Presburger arithmetic (\SL{}, in \Cref{sec:SL})
		that allows expressing properties about summations.
		We show how we can formalize verification problems of smart contracts in \SL{}.
	\item
		We discuss the decidability problem of checking
                validity of  \SL{} formulas
		 (\Cref{sec:decidability}):
		we prove that it is undecidable in the general case,
		but also that there exists a small decidable fragment.
		
	\item
		We show different encodings of \SL{} to first-order 
		logic (\Cref{sec:encoding}).
        To this end, we consider 
        theory-specific reasoning and variations of \SL{}, for
        example by replacing non-negative integer reasoning
        with term algebra properties.
                
	\item We evaluate our results with SMT solvers  and
	first-order theorem provers, by using 31 new benchmarks
	encoding smart transitions and their properties (\Cref{sec:exp}). 
	Our experiments demonstrate the applicability of our results
	within automated reasoning, in a fully automated manner,
	without any user guidance.
	
\end{itemize}

\section{Preliminaries}

We consider many-sorted first-order logic (FOL) with equality, 
defined in the standard way. 
The equality symbol is denoted by $\approx$. 

We denote by $\structure{ \Sigma }$ the 
\emph{set of all structures} for the vocabulary $\Sigma$.
A structure $\A \in \structure{\Sigma}$ is
a pair $\pair{\domain}{\interp}$,
where for each sort $\sort{s}$, its domain in $\A$ is $\domain(\sort{s})$,
and for each symbol $S$, its interpretation in $\A$ is $\interp(S)$.
Note that \emph{models} of a formula $\varphi$ over a vocabulary
$\Sigma$ are structures $\A \in \structure{\Sigma}$. 

A \emph{first-order theory}  is a set of first-order formulas closed
under logical consequence.   We will consider, the
first-order theory of the natural numbers with addition. This is Presburger arithmetic (PA) which is
of course decidable~\cite{presburger-decidable}. \\
We write $\N$ to denote the set of natural numbers.
We consider $0 \in \N$ and write $\N^+$ to explicitly exclude $0$ from
$\N$. The vocabulary of PA is 
$
\Sigma_\text{Presburger} = \parens{ 0, 1, c_1, \dots, c_l, +^2 }
$, with all constants $0, 1, c_i$ of sort $\Nat$.
A structure $\A = (\domain, \interp) \in \structure{\Sigma_\text{Presburger}}$ is called 
a \emph{Standard Model of Arithmetic} when 
$\domain(\Nat) = \N$ and $+^2$ is interpreted as the standard
binary addition $+$ function over the naturals.  
The vocabulary $\Sigma_\text{Presburger}$  can be extended with a
total order relation, yielding 
$\Sigma^*_\text{Presburger} = \parens{ 0, 1, +^2, \leq^2 }$, 
where $\leq^2$ is interpreted as the binary relation $\leq$ in Standard Models of Arithmetic.

\section{Sum Logic (\SL{})}\label{sec:SL}
We now define \emph{Sum Logic (\SL{})} as a generalization of Presburger
arithmetic, extending Presburger arithmetic with unbounded sums. 
\SL{} is motivated by applications of financial transactions over
cryptocurrencies in smart contracts. 
Smart contracts are decentralized computer programs executed on 
a blockchain-based system, as explained in \cite{smart}. 
Among other tasks, they automate financial transactions such as 
transferring and minting money. 
We refer to these transactions as \emph{smart transitions}.
The aim of this paper and \SL{} in particular is to express and reason about 
the post-conditions of smart transitions similar to \Cref{fig:mint-n-high-level}.

\SL{} expresses smart transition relations among sums of accounts of various kinds, 
e.g., at different banks, times, etc. 
Each such kind, $j$, 
is modeled by an uninterpreted function symbol, $b_j$, 
where $b_j(a)$ denotes the balance of $a$'s account of kind $j$,
and a constant symbol $s_j$,
which denotes the sum of all outputs of $b_j$. 
As such, our \SL{} generalizes Presburger arithmetic with 
(i) a sort $\Addr$ corresponding to the (unbounded) set of 
account \textit{addresses}; 
(ii) \textit{balance} functions $b_j$ mapping account addresses from
$\Addr$ to account values of sort $\Nat$; and
(iii) \textit{sum constants $s_j$} of sort $\Nat$ capturing 
the total sum of all account balances represented by $b_j$. 
Formally, the vocabulary of \SL{} is defined as follows.

\begin{definition}[\SL{} Vocabulary]
	\label{def:sum-vocabulary}
	Let 
	\[
	\Sigma^{l,m,d}_{+,\leq} = \parens{ 
		a_1, \dots, a_l, 
		b^1_1, \dots, b^1_m,
		c_1, \dots, c_d,
		s_1, \dots, s_m,
		0, 1,
		+^2, \leq^2
	}
	\]
	be a \emph{sorted first-order vocabulary of \SL{}} over sorts $\braces{\Addr, \Nat}$, where 
	
	\begin{itemize}
		\item (Addresses) The constants $a_1, \dots, a_l$ are of sort
		$\Addr$; 
		\item (Balance functions) 
		$b^1_1, \dots, b^1_m$ are unary function
		symbols
		from $\Addr$ to $\Nat$; 
		\item  (Constants and Sums) The constants $c_1, \dots, c_d,
		s_1, \dots, s_m$ and $0, 1$ are of sort $\Nat$; 
		\item $+^2$ is a binary function $\Nat \times \Nat \to
		\Nat$; 
		\item $\leq^2$ is a binary relation over $\Nat \times
		\Nat$. 
	\end{itemize}
\end{definition}

In what follows, when the cardinalities in an \SL{}
vocabulary are clear from context, we simply write $\Sigma$ instead of $\Sigma^{l,m,d}_{+,\leq}$.
Further, by $\Sigma^{l,m,d}_{\cancel{+},\cancel{\leq}}$ we denote the sub-vocabulary
where the crossed-out symbols are not available.
Note that even when addition is not available,
we still allow writing numerals larger than 1.

We restrict ourselves to \emph{universal sentences} over an \SL{} vocabulary, 
with quantification only over the \Addr{} sort.

We now extend the Tarskian semantics of first-order logic to
ensure that the sum constants of an \SL{} vocabulary
($s_1, \dots, s_m$)
are equal to the sum of outputs of their associated balance functions 
($b_j$ for each $s_j$)
over the respective entire domains of sort $\Addr$.

Let $\Sigma$ be an \SL{} vocabulary.
An \SL{} structure $\A = \pair{\domain}{\interp} \in \structure{\Sigma}$
representing a model for an \SL{} formula $\varphi$ is called an \emph{\SL{} model}
iff
\[\tag{Sum Property}\label{eq:sum-property}
\interp(s_j) = \sum_{a \in \domain(\Addr)} \brackets{\interp(b_j)}(a), \quad \text{for each }
1\leq j\leq m.
\]

We write 
$\A \sat_{\SL{}} \varphi$ 
to mean that $\A$ is an \SL{} model of $\varphi$.
When it is clear from context, we simply write 
$\A \sat \varphi$. 

\begin{table}[tb]
	\centering
	\begin{tabular}{| l | c | l |}
		\hline
		\multicolumn{1}{|c|}{\textbf{Function}} & 
		\multicolumn{1}{c|}{\textbf{Encoding in \SL{}}} &
		\multicolumn{1}{c|}{\textbf{$\,$Reference in ERC-20$\,$}}
		\\ \hline
		$\,$\msum{} & $s$ or $s'$ & $\,$totalSupply
		\\ \hline
		$\,$\bal$(a)$ & $b(a)$ or $b'(a)$ & $\,$balanceOf
		\\ \hline
		$\,$\mint$(a, v)$ & $b'(a) \approx b(a) + v$ & $\,$transfer
		\\ \hline
		$\,$\texttt{transferFrom}$(f, t, v)$$\,$ & $\,$$b'(t) \approx b(t) + v \land b(f) \approx b'(f) + v $$\,$ & $\,$transferFrom
		\\
		\hline
	\end{tabular}
	\vspace*{2mm}
	\caption{ERC-20 Token Standard}
	\label{tbl:sum-logic-erc-20} \vspace{-0.7cm}
\end{table}

\begin{example}[Encoding ERC-20 in \SL{}]
As a use case of $\SL{}$, we showcase the encoding of the 
ERC-20 token standard of the Ethereum community~\cite{erc-20-standard} in \SL{}. 
To this end, we consider an \SL{} vocabulary $\Sigma^{l,2,d}$. 
We respectively denote the balance functions and their associated sums as 
$b, b',s, s'$ in the \SL{} structure over $\Sigma^{l,2,d}$. 
The resulting instance of \SL{} can then be used to encode ERC-20
operations/smart transitions as
\SL{} formulas, as shown in \Cref{tbl:sum-logic-erc-20}.
Using this encoding, the post-condition of
\Cref{fig:mint-n-high-level} is expressed as the \SL{}
formula
\begin{equation}\label[formula]{eq:mint-n-sum-logic}
	b'(a) \approx b(a) + n \;\land \;\forall a' \not\approx a. b'(a') \approx b(a')
	\; \land \; s' \approx s + n
\end{equation}
formalizing the correctness of the smart transition of minting $n$
tokens in \Cref{fig:mint-n-high-level}.  
%
In the applied verification examples in \Cref{sec:exp}, 
rather than verifying the low-level implementation 
of built-in functions such as $\mint_n$,
we assume their correctness by including suitable axioms.

\end{example}

\section{Decidability of \SL{}}\label{sec:decidability}
We consider the decidability problem of verifying formulas in \SL{}.
We show that when there are several function symbols $b_j$ to sum over, the
  satisfiability problem for \SL{} becomes undecidable\footnote{Due to
    space restrictions, proofs of our results are given in our Appendix.}.  We first present, however, a useful
  decidable fragment of \SL{}.

\subsection{A Decidable Fragment of \SL{}}\label{sec:decidable-sum-logic}
We  prove decidability for a fragment of \SL{}, which we call the
$\parens{l,1,d}$-\textsc{FRAG} fragment of \SL{} (\Cref{thm:sum-decidability}). 
For doing so, we reduce the fragment to Presburger arithmetic, 
by using regular Presburger constructs to 
encode \SL{} extensions, 
that is the uninterpreted functions and sum
constants of \SL{}. 

The first step of our reduction proof is to consider distinct
models, which are models where the $\Addr$ constants $a_i$ represent distinct elements in the domain \linebreak $\domain(\Addr)$. 
While this restriction is somewhat unnatural, we show that for each vocabulary and formula that has a model, 
there exists an equisatisfiable formula over a different vocabulary 
that has a \emph{distinct} model (\Cref{thm:distinct-models}). 
The crux of our decidability  proof is then proving that
$\parens{l,1,d}$-\textsc{FRAG} has \emph{small \Addr{} space}: 
given a formula $\varphi$, if it is satisfiable, 
then there exists a model where $\abs{ \domain(\Addr) } \leq \kappa(\abs{\varphi})$, 
 $\abs{\varphi}$ is the length of $\varphi$, 
and $\kappa(.)$ is some computable function (\Cref{thm:presburger-reduction})\footnote{
	The function $\kappa(.)$ is defined per decidable fragment of SL,
	and not per formula.
      }. 

\vspace*{-5.5mm}
\subsubsection{Distinct Models}
An \SL{} structure $\A$ is considered \emph{distinct} when the $l$ \Addr{} constants represent 
$l$ distinct elements in $\domain(\Addr)$.
I.e.,
\vspace*{-2.5mm}
\[
\abs{ \braces{ \interp(a_1), \dots, \interp(a_l) } } = l\,.
\]
\vspace*{-5.5mm}

\noindent
Since each \SL{} model induces an equivalence relation over the $\Addr$ constants,
we consider partitions $P$ over $\braces{a_1, \dots, a_l}$. For each possible partition $P$ we define a transformation of terms and formulas
$\mathcal{T}_P$ that substitutes equivalent $\Addr$ constants
with a single $\Addr$ constant.
The resulting formulas are defined over a vocabulary that has
$\abs{P}$ $\Addr$ constants.  We show that given an \SL{} formula $\varphi$,
	if $\varphi$ has a model,
	we can always find a partition $P$ such that each of its classes corresponds to
	an equivalence class induced by that model.

\begin{restatable}[Distinct Models]{theorem}{distinctModelsTheorem}\label{thm:distinct-models}
	Let $\varphi$ be an \SL{} formula over $\Sigma$, then
	$\varphi$ has a model 
	iff
	there exists a partition $P$ of $\braces{a_1,\dots,a_l}$ 
	such that $\mathcal{T}_P(\varphi)$ has a \emph{distinct} model.\qed
\end{restatable}

\vspace*{-7.5mm}
\subsubsection{Small \Addr{} Space}
In order to construct a reduction to Presburger arithmetic,
we bound the size of the $\Addr$ sort.
For a fragment of \SL{} to be decidable, 
we therefore need a way to bound its models upfront.
We formalize this requirement as follows.

\begin{restatable}[Small \Addr{} Space]{definition}{smallModelProperty}
	\label{def:small-model-property}
	Let \textsc{FRAG} be some fragment of SL over vocabulary 
	$\Sigma = \Sigma^{l,m,d}_{+,\leq}$. 
	\textsc{FRAG} is said to have \emph{small \Addr{} space} 
	if there exists a computable function $\kappa_\Sigma(.)$, 
	such that for any \SL{} formula $\varphi \in \textsc{FRAG}$, 
	$\varphi$ has a distinct model 
	 iff
	$\varphi$ has a distinct model $\A = (\domain, \interp)$ 
	with \emph{small \Addr{} space},  
	where $\abs{\domain(\Addr)} \leq
        \kappa_\Sigma(\abs{\varphi})$.
        
	We call $\kappa_\Sigma(.)$ the  \emph{bound function} of \textsc{FRAG};
	when the vocabulary is clear from context we simply write $\kappa(.)$.
\end{restatable}

One instance of a fragment 
(or rather, family of fragments) 
that satisfies this property
is the $\parens{l,1,d}$-\textsc{FRAG} fragment:
the simple case of a \emph{single} uninterpreted ``balance'' function
(and its associated sum constant),
further restricted by removing the binary function $+$ and
the binary relation $\leq$.
Therefore, we derive the following theorem:

\begin{restatable}[Small \Addr{} Space of $\parens{l,1,d}$-\textsc{FRAG}]{theorem}{smallModelsForSingle}
	\label{thm:small-models-for-single}
	
	For any $l$, $d$, it holds $\parens{l,1,d}$-\textsc{FRAG}, the fragment of 
	\SL{} formulas over the \SL{} vocabulary 
	\[
	\Sigma^{l,1,d}_{\cancel{+},\cancel{\leq}} 
		= 
	\parens{ a_1, \dots, a_l, b^1, c_1, \dots, c_d, s, 0, 1 } \;,
	\]
	has \emph{small \Addr{} space} with bound function
	$\kappa(x) = l + x + 1$.\qed
\end{restatable}

	An attempt to trivially extend
        \Cref{thm:small-models-for-single} 
	for a fragment of \SL{} with two balance functions falls apart
    in a few places,
	but most importantly when comparing balances to the sum of a different balance function.
	%
In \Cref{sec:undecidable} we show that  these comparisons are essential 
	for proving our undecidability result in \SL{}.
	

\subsubsection{Presburger Reduction}
For showing decidability of some \textsc{FRAG} fragment of \SL{}, 
we describe a Turing reduction to pure Presburger arithmetic. 
We introduce a transformation $\tau(.)$ of formulas in
\SL{}  
into formulas in Presburger arithmetic. It maps universal quantifiers to disjunctions,
and sums to explicit addition of all balances.
In addition, we define an auxiliary formula $\eta(\varphi)$, which ensures only valid addresses are considered,
and that invalid addresses have zero balances.
The formal definitions of $\tau(.)$ and $\eta(\varphi)$ can be found in \Cref{app:sum-logic-proofs}.

By relying on the properties of \emph{distinctness} and \emph{small \Addr{} space} 
we get the following results.

\begin{restatable}[Presburger Reduction]{theorem}{presburgerReductionTheorem}
	\label{thm:presburger-reduction}
	An \SL{} formula $\varphi$ has a \emph{distinct}, \SL{} model with small \Addr{} space
		iff
	$\tau(\varphi) \land \eta(\varphi)$ 
	has a Standard Model of Arithmetic.\qed
\end{restatable}

\begin{theorem}[\SL{} Decidability]\label{thm:sum-decidability}
	Let \textsc{FRAG} be a fragment of \SL{} that has \emph{small \Addr{} space},
	as defined in \Cref{def:small-model-property}.
	Then, \textsc{FRAG} is decidable.
\end{theorem}

\begin{proof}[\Cref{thm:sum-decidability}]
	Let $\varphi$ be a formula in \textsc{FRAG}. Then
	$\varphi$ has an \SL{} model 
		 iff
	for some partition $P$ of $\braces{a_1,\dots,a_l}$, 
	$\mathcal{T}_P(\varphi)$ has a \emph{distinct} \SL{} model. 
	For any $P$, the formula $\mathcal{T}_P(\varphi)$ is in FRAG, 
	therefore 
    $\mathcal{T}_P(\varphi)$ has a \emph{distinct} \SL{} model 
		 iff
	it has a \emph{distinct} \SL{} model with \emph{small \Addr{} space}.
	
	From \Cref{thm:presburger-reduction}, 
	we get that for any $P$, 
	$\varphi_P \define \mathcal{T}_P(\varphi)$ has a \emph{distinct} \SL{} model 
		 iff
	$\tau(\varphi_P) \land \eta(\varphi_P)$ has a Standard Model of Arithmetic.
	By using the PA decision procedure as an oracle, we obtain the
        following {\emph{decision procedure for a \textsc{FRAG} formula
        $\varphi$}}: 
	\begin{itemize}
		\item For each possible partition $P$ of $\braces{a_1,\dots,a_l}$,
		 let $\varphi_P = \mathcal{T}_P(\varphi)$; 
		
		\item Using a PA decision procedure,  check whether 
		$\tau(\varphi_P) \land \eta(\varphi_P)$ 
		has a model, for each $P$; 
		
		\item If a model 
                  for some partition $P$ was found, 
		the formula $\varphi_P$ has a \emph{distinct} \SL{} model, 
		and therefore $\varphi$ has \SL{} model; 
		
		\item Otherwise, there is no \emph{distinct} \SL{} model for any partition $P$, 
		and therefore there is no \SL{} model for $\varphi$.
	\end{itemize}
      \end{proof}

\begin{remark}
Our decision procedure 
for \Cref{thm:sum-decidability} requires $B_l$ Presburger queries, \linebreak
where $B_l$ is Bell's number for all possible partitions of a set of size $l$. 
\end{remark}

Using \Cref{thm:sum-decidability} and \Cref{thm:small-models-for-single}, 
we then obtain the following result. 
\begin{corollary}
$\parens{l,1,d}$-\textsc{FRAG} is decidable.\qed
\end{corollary}


\subsection{\SL{} Undecidability}\label{sec:undecidable}
We now show that simple extensions of our decidable 
$\parens{l,1,d}$-\textsc{FRAG} fragment lose its decidability (\Cref{thm:undecidable}).
For doing so, we encode the halting problem of a \twocm{} 
using \SL{} with 3 balance functions, 
thereby proving that the resulting \SL{} fragment is undecidable. 

Consider a \twocm{}, 
whose transitions are encoded by the Presburger formula 
$\pi(c_1, c_2, p, c'_1, c'_2, p')$ with 6 free variables: 
2 for each of the three registers, 
one of which being the program counter (\pc{}). 
We assume w.l.o.g. that all three registers are within
$\N^+$, allowing us 
to use addresses with a zero balance as a special ``separator''.
In addition, we assume that the program counter is 1 
at the start of the execution, 
and that there exists a single halting statement at line $H$. 
That is, the \twocm{} halts iff the \pc{} is equal to $H$.
\begin{table}[tb]
	\centering
	\bgroup
	\setlength\tabcolsep{2mm}
	\begin{tabular}{r | c | c | c | l |}
		\cline{2-5}
		& \textbf{\Addr} & \textbf{$l(\Addr)$} & \textbf{$c(\Addr)$} & \multicolumn{1}{c|}{\textbf{$g(\Addr)$}}
		\\ \cline{2-5}
		\ldelim\{{4}{*}[\mbox{{\footnotesize Time-step  \#0}}] & 
		& 0 & 1 & 0
		\\ \cline{2-5}
		& & 1 & 1 & $c_1$ at \#0
		\\ \cline{2-5}
		& & 2 & 1 & $c_2$ at \#0
		\\ \cline{2-5}
		& $a_0$ & 3 & 1 & \textsc{pc} at \#0 = 1
		\\ \cline{2-5}
		& \vdots & \vdots & \vdots & \multicolumn{1}{c|}{\vdots}
		\\ \cline{2-5}
		\ldelim\{{4}{*}[\mbox{{\footnotesize Time-step   \#$i$}}] & $x_1$ & $4i$ & 1 & 0
		\\ \cline{2-5}
		& $x_2$ & $4i + 1$ & 1 & $c_1$ at \#$i$
		\\ \cline{2-5}
		& $x_3$ & $4i + 2$ & 1 & $c_2$ at \#$i$
		\\ \cline{2-5}
		& $x_4$ & $4i + 3$ & 1 & \textsc{pc} at \#$i$
		\\ \cline{2-5}
		\ldelim\{{4}{*}[\mbox{{\footnotesize Time-step \#$(i+1)$}}] & $x_5$ & $4i + 4$ & 1 & 0
		\\ \cline{2-5}
		& $x_6$ & $4i + 5$ & 1 & $c_1$ at \#$(i+1)$
		\\ \cline{2-5}
		& $x_7$ & $4i + 6$ & 1 & $c_2$ at \#$(i+1)$
		\\ \cline{2-5}
		& $x_8$ & $4i + 7$ & 1 & \textsc{pc} at \#$(i+1)$
		\\ \cline{2-5}
		& \vdots & \vdots & \vdots & \multicolumn{1}{c|}{\vdots}
		\\ \cline{2-5}
		\ldelim\{{4}{*}[\mbox{{\footnotesize Time-step   \#$n = \frac{s_c}{4} - 1$}}] & & $s_c - 4$ & 1 & 0
		\\ \cline{2-5}
		& & $s_c - 3$ & 1 & $c_1$ at \#$n$
		\\ \cline{2-5}
		& & $s_c - 2$ & 1 & $c_2$ at \#$n$
		\\ \cline{2-5}
		& $a_1$ & $s_c - 1$ & 1 & \textsc{pc} at \#$n$ = $H$
		\\ \cline{2-5}
	\end{tabular}
	\egroup
	\vspace*{3mm}
	\caption{Transition System of a 2-Counter Machine,   Array View.}\label{tbl:2cm-in-sl}
	\vspace{-0.7cm}
\end{table}

\vspace*{-3.5mm}
\subsubsection{Reduction Setting}
We have 4 \Addr{} elements for each time-step, 
3 of them hold one register each, 
and one is used to separate between each group of \Addr{} elements
(see \Cref{tbl:2cm-in-sl}).
We have 3 uninterpreted functions from \Addr{} to \Nat{} (``balances'').
For readability we denote these functions as $c,l,g$ (instead of $b_1,b_2,b_3$)
and their respective sums as $s_c,s_l,s_g$:
\begin{enumerate}
	\item Function $c\,$: Cardinality function, used to force size constraints. 
	We set its value for all addresses to be $1$, 
	and therefore the number of addresses is $s_c$.
	
	\item Function $l\,$: Labeling function, to order the time-steps. 
	We choose one element to have a maximal value of $s_c - 1$ 
	and ensure that $l$ is injective. 
	This means that the values of $l$ are distinctly $[0, s_c - 1]$.
	
	\item Function $g\,$: General purpose function, 
	which holds either one of the registers or 0 to mark the \Addr{} element as a separating one.
\end{enumerate}

\noindent
Each group representing a time-step is a 4 \Addr{} element, ordered as follows:
\begin{enumerate}
	\item First, a separating \Addr{} element $x$ (where $g(x)$ is 0).
	\item Then, the two general-purpose counters.
	\item Lastly, the program counter.
\end{enumerate}

\noindent
In addition we have 2 \Addr{} constants, 
$a_0$ and $a_1$ which represent the \pc{} value 
at the start and at the end of the execution. 
The element $a_1$ also holds the maximal value of $l$, that is,
$l(a_1) + 1\approx s_c$. 
Further, $a_0$ holds the fourth-minimal value, 
since its the last element of the first group,
and each group has four elements.

\vspace*{-3.5mm}
\subsubsection{Formalization Using a \TwoCM} 
We now formalize our reduction, proving undecidability of \SL{}.


\noindent 
(i) We impose an injective labeling 
	\vspace*{-2mm}
	\[
	\varphi_1 = 
		\forall x,y. 
			\parens{ l(x) \approx l(y) } 
				\to 
			\parens{ x \approx y }
	\]
	\vspace*{-5mm}

\noindent 
(ii) We next formalize properties over the program counter \pc{}.
       The \Addr{} constant that represents the program counter
       \pc{} value of the 
	last time-step is set to have the maximal labeling, 
	that is
	\vspace*{-2mm}
	\[
	\varphi_2 = \forall x. l(x) \leq l(a_1)
	\]
        Further, the \Addr{} constant that represents the \pc{} value of the 
	first time-step has the fourth labeling, hence
	\vspace*{-3mm}
	\[
	\varphi_3 = l(a_0) \approx 3
	\]
	Finally, the first and last values of the program counter are
        respectively 1 and
        $H$,  that is 
    \vspace*{-3mm}
	\[
	\varphi_4 = 
		g(a_0) \approx 1 
			\land 
		g(a_1) \approx H
	\]

\noindent 
(iii) We express \emph{cardinality constraints}
	ensuring that 
	there are as many \Addr{} elements as the labeling of the 
	last \Addr{} constant ($a_1$) + 1. 
	We assert 
	\vspace*{-3mm}
	\[
	\varphi_5 = 
		\parens{ s_c \approx l(a_1) + 1 } 
			\land 
		\forall x. \parens{ c(x) \approx 1 }
	\]
	
\noindent 
(iv) We encode the transitions of the \twocm{}, 
as follows. 
	For every 8 \Addr{} elements, 
	if they represent two sequential time-steps, 
	then the formula for the transitions of the \twocm{}
        is valid 
	for the registers it holds. 
	As such, we have \vspace{-0.3cm}
	\begin{align*}
		\varphi_6 = 
			\forall &x_1,\dots,x_8.
				\parens{ F1 \land F2 \land F3 } 
				\\
					&\to 
				\pi \parens{
					g(x_2), g(x_3), g(x_4), g(x_6), g(x_7), g(x_8)
				}
	\end{align*}
    where the conjunction $F1 \land F2 \land F3$ expresses that 
    $x_1, \dots, x_8$ are two sequential time-steps, with $F1$, $F2$ and $F3$
    defined as below. In particular, 
    $F1$, $F2$ and $F3$ formalize that  $x_1, \dots, x_8$  have sequential labeling, 
    starting with one zero-valued \Addr{} element
        (``separator'') and 
	continuing with 3 non-zero elements, as follows: 
	\begin{itemize}
		\item Sequential: 
		\vspace*{-4mm}
		\[\tag{F1}
			l(x_2) \approx l(x_1) + 1 
				\land \dots \land 
			l(x_8) \approx l(x_7) + 1
		\]
		\item Time-steps:
		\vspace*{-8mm}
		\begin{align*} &\tag{F2}
			 g(x_1) \approx 0 
				\land 
			g(x_2) > 0 
				\land 
			g(x_3) > 0 
				\land 
			g(x_4) > 0 \;, \\
	&\tag{F3}
			g(x_5) \approx 0 
				\land 
			g(x_6) > 0 
				\land 
			g(x_7) > 0 
				\land 
			g(x_8) > 0
		\end{align*}
	\end{itemize}

Based on the above formalization, the formula 
$\varphi = \varphi_1 \land \dots \land \varphi_6$ is satisfiable 
 iff
the \twocm{} halts within a finite amount of time-steps 
(and the exact amount would be given by $\frac{s_c}{4}$).
Since the halting problem for \twocms{} is undecidable, 
our \SL{}, 
already with 3 uninterpreted functions and their associated sums, 
is also undecidable.

\begin{theorem}\label{thm:undecidable}
	For any $l \geq 2, m \geq 3$ and $d$, 
	any fragment of \SL{} over
	$\Sigma^{l,m,d}_{+,\leq}$
	is undecidable. \qed
\end{theorem}

\vspace*{-5mm}
\begin{remark}
Note that in the above formalization
the only use of associated sums comes from 
expressing the size of the set of \Addr{} elements. 
As for our uninterpreted function $c(.)$ we have 
$\forall x. c(x) \approx 1$, 
its sum $s_c$ is thus the amount of addresses. 
Hence, 
we can encode the halting problem for \twocms{} in an
almost identical way to the encoding presented here, 
using a generalization of PA with two uninterpreted functions 
for $l(.)$ and $g(.)$, 
and a \emph{size operation} replacing $c(.)$ and its
associated sum.
\end{remark}

\section{\SL{} Encodings of Smart Transitions} \label{sec:encoding}
The definition of \SL{} models in
\Cref{sec:SL,sec:decidability} ensured that
the summation constants $s_j$ were respectively equal to 
the actual summation of all balances $b_j(.)$. 
In this section, we address the
challenge to formalize relations between $s_j$ and
$b_j(.)$ in a way that the resulting encodings can be expressed in
the logical frameworks of automated reasoners, 
in particular of SMT solvers and first-order theorem provers. 

In what follows,
we consider a single transaction or one time-step of
multiple transactions over $s_j, b_j(.)$.
We refer to such transitions as \emph{smart transitions}.
Smart transitions are common in smart contracts,
expressing for example the minting and/or transferring of some coins,
as evidenced in \Cref{fig:mint-n-high-level} and discussed later.

Based on \Cref{sec:SL},
our smart transitions are encoded in the
$\Sigma^{l,2,d}$ fragment of \SL{}.
Note however, that neither decidability nor undecidability
of this fragment is implied by \Cref{thm:sum-decidability},
nor \Cref{thm:undecidable}. 
In this section,
we show that our \SL{} encoding of smart transitions is expressible in
first-order logic. 
We first introduce a sound, \emph{implicit \SL{} encoding},
by ``hiding'' away sum semantics and using invariant relations
over smart transitions (\Cref{sec:enc_uf}). 
This encoding does not allow us to directly assert
the values of any balance or sum,
but we can prove that this implicit encoding is complete,
relative to a translation function (\Cref{sec:enc_sound}).

By further restricting our implicit \SL{} encoding to
this relative complete setting,
we consider counting properties to
explicitly reason with balances 
and directly express verification conditions with unbounded sums on 
$s_j$ and $b_j(.)$.
This is shown in \Cref{sec:expl_enc},
and we evaluate different variants of the 
\emph{explicit \SL{} encoding} in \Cref{sec:exp},
showcasing their practical use and relevance within
automated reasoning.

To directly present our \SL{} encodings and results in the smart contract
  domain, in what follows we rely on the notation of 
\Cref{tbl:sum-logic-erc-20}. As such, we respectively denote $b,b'$ by 
\oldbal{}, \newbal{} and write \oldsum{}, \newsum{} for
 $s,s'$. As already discussed in  \Cref{fig:mint-n-high-level}, 
the prefixes \texttt{old-} and \texttt{new-} refer to
the entire state expressed in the encoding before and after the smart transition.
We explicitly indicate this state using \oldworld{}, \newworld{} respectively.
The non-prefixed versions \bal{} and \msum{} are stand-ins for
\emph{both} the \old{} and \new{} versions --- \Cref{fig:enc_uf}
illustrates our setting for  the smart transition of minting one coin. 


With this \SL{} notation at hand, 
we are thus interested in
finding first-order formulas that
verify smart transition relations between \oldsum{} and \newsum{},
given the relation between \oldbal{} and \newbal{}. 
In this paper, we mainly focus on the smart transitions
of minting and transferring money, 
yet our results could be used in the context
of other financial transactions/software transitions over unbounded sums.

\begin{example}
In the case of minting $n$ coins in  \Cref{fig:mint-n-high-level},
	we require formulas that 
	(a) 
		describe the state before the transition 
		(the \oldworld{}, thus pre-condition),
	(b) 
		formalize the transition 
		(the relation between \oldbal{} and \newbal{}; 
		(i)-(ii) in \Cref{fig:mint-n-high-level})
	and (c) 
		imply the consequences for the \newworld{} 
		((iii) in \Cref{fig:mint-n-high-level}). 
	These formulas verify that minting and depositing $n$ coins
	into some address
	result in an increase of the sum by $n$, 
	that is $\newsum = \oldsum + n$, 
	as expressed in the functional correctness formula
	(\ref{eq:mint-n-sum-logic}) of \Cref{fig:mint-n-high-level}.
\end{example}


\subsection{\SL{} Encoding using Implicit Balances and Sums}\label{sec:enc_uf}
The first encoding we present is a set of first-order formulas with 
equality over sorts $\braces{ \Coin,\, \Addr }$. No additional theories
are considered. 
The \Coin{} sort represents money, 
where one coin is one unit of money.
The \Addr{} sort represents the account addresses as before.
As a consequence, 
balance functions and sum constants only exist implicitly
in this encoding. As such, 
the property 
$\msum=\sum_{a \in \Addr} \bal(a)$
\emph{cannot be directly expressed in this encoding}.
Instead, we formalize this property by using so-called \emph{smart
  invariant} relations between two predicates \hc{} and
\ac{} over coins $c\in\Coin$ and $a\in\Addr$, as follows.

%

\begin{figure}[tb]
	\centering
	\includegraphics[width=0.9\textwidth]{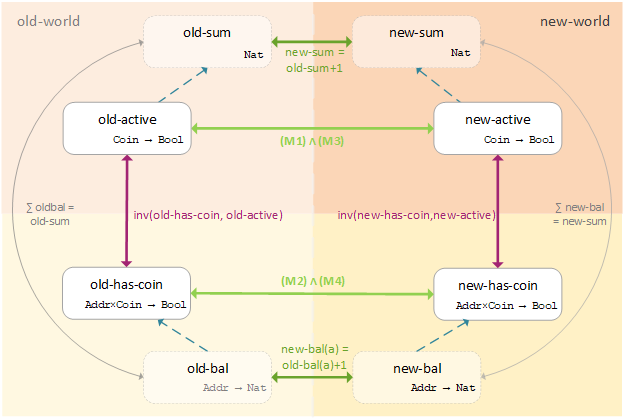}
	\caption{Implicit \SL{} Encoding of $\mint_1$, where \texttt{Addr} is short for \Addr.}
	\label{fig:enc_uf}
\end{figure}

\begin{definition}[Smart Invariants]\label{def:invariants}
Let $\hc\subseteq \Addr{} \times \Coin{}$ 
	and  consider $\ac \subseteq \Coin{}$. A \emph{smart invariant} of
        the pair $\pair{\hc}{\,\ac}\,$ is
        the conjunction of the following three formulas 
  %
	\begin{enumerate}
		\item 
			Only active coins $c$ can be owned by an address $a$:
			\[ 
			\tag{I1} 
			\forall c : \Coin.\;
			\exists a:\Addr. \;
			\hc(a,c) \to \ac(c) 
			\;.
			\]
		\item 
			Every active coin $c$ belongs to some address $a$:
			\[ 
			\tag{I2} 
			\forall c : \Coin.\;
			\ac(c) \to 
			\exists a : \Addr. \;
			\hc(a,c) 
			\;.
			\]
		\item  
			Every coin $c$ belongs to at most one address $a$: 
			\begin{align*}
			\tag{I3}  
			\forall c : \Coin.
			&\forall a, a' : \Addr. \\
			&\parens{ 
				\hc(a,c) \land \hc(a',c) \to a \approx a'
			}.
			\end{align*}
	\end{enumerate}
We write  $\inv(\hc{},\;\ac{})$ to denote the smart invariant $\text{(I1)} \land
        \text{(I2)} \land \text{(I3)}$ of $\pair{\hc}{\,\ac}\,$. 
      \end{definition}

      Intuitively, our \emph{smart invariants}  ensure that
a coin $c$ is \emph{active} iff it is \emph{owned} by precisely one
address $a$.
%
%
%
Our smart invariants imply the soundness of our implicit \SL{} encoding,
as follows.

\begin{restatable}[Soundness of \SL{} Encoding]{theorem}{soundImplicitSL}
\label{thm:enc:soundness}
  Given that
$\msum = \abs{\ac}$
and for every $a \in \Addr$ it holds 
$\bal(a) = \abs{ \braces{ c \in \Coin \mid \pair{a}{c} \in \hc } }$, then
$	\inv(\hc,\ac)\; \Longrightarrow \;
  \msum = \suma \bal(a)$. \qed

\end{restatable} 


\noindent We say that a \emph{smart transition preserves smart invariants},
when
\begin{align*}
	\inv&(\oldhc, \oldact) \\
	&\iff \inv(\newhc,\newact),
\end{align*}
 
  where $\oldhc, \oldact$ and $\newhc, \newact$
  respectively denote the functions $\hc, \ac$ in the states before
  and after the smart transition. 
  Based on the soundness of our implicit \SL{} encoding, we formalize smart
  transitions preserving smart invariants as first-order formulas. 
  We only discuss  smart transitions implementing 
minting $n$ coins here, but other transitions, such as transferring coins, can be handled
in a similar manner. We first focus on miniting a
single coin, as 
 follows. 
	%
	%


\begin{definition}[Transition $\mint_1(a,c)$]\label{mint1}
	Let there be $c \in \Coin, a \in \Addr$. 
	The transition $\mint_1(a,c)$ activates coin $c$ 
	and deposits it into address $a$.
	\begin{enumerate}
		\item 	
			The coin $c$ was inactive before and is active now: 
			\[ 
			\tag{M1} 
			\neg \oldact(c) \land \newact(c) \;.
			\]
			
		\item 
			The address $a$ owns the new coin $c$:
			\[ 
			\tag{M2}  
			\newhc(a,c) \land \forall a' : \Addr{}.\; \neg \oldhc(a',c)\;.
			\]
			
		\item 
			Everything else stays the same: 
			\begin{align*}
			\tag{M3}
			& \forall  c' : \Coin.\; 
				c' \not\approx c\to 
				\parens{ \newact(c') \liff \oldact(c') }\;, \\
				\tag{M4}
		& \forall  c' : \Coin.\;\forall a' : \Addr. \;
		\; 
		(c' \not\approx c \lor a' \not\approx a) \to \\
		& 	\qquad  \parens{ \newhc(a',c') \liff \oldhc(a',c') } \;.
			\end{align*}
	\end{enumerate}
	The transition $\mint_1(a,c)$ is defined as 
	$\text{(M1)} \land \text{(M2)} \land \text{(M3)} \land \text{(M4)}$.
\end{definition}
 

By minting one coin, 
	the balance of precisely one address, that is of the
        receiver's address, increases by one,
	whereas all other balances remain unchanged. Thus, 
	the expected impact on the sum of account balances is also
        increased by one, as illustrated in \Cref{fig:enc_uf}. 
The following theorem proves that the definition of 
$\mint_1$ is \emph{sound}. That is, $\mint_1$ 
 affects the implicit balances and sums as expected and hence $\mint_1$
 preserves smart invariants. 

\begin{restatable}[Soundness of $\mint_1(a,c)$]{theorem}{soundMint}\label{mint1_sound}
	Let  $c \in \Coin$, $a\in \Addr$ such that 
        $\mint_1(a,c)$. 
	Consider balance functions
	$\oldbal$, $\newbal: \Addr \to \N$,
	non-negative integer constants $\oldsum$, $\newsum$,
	unary predicates $\oldact$, $\newact \subseteq \Coin$
	and binary predicates $\oldhc$, $\newhc \subseteq \Addr \times \Coin$
	such that 
	\[
		\abs{\oldact}=\oldsum\, ,\;
		\abs{\newact}=\newsum , 
	\] 
	and for every address $a'$, we have 
	\begin{align*}
		&\oldbal(a')=\abs{ \braces{ c' \in \Coin \mid (a', c') \in \oldhc } } \, ,\\ 
		&\newbal(a')=\abs{ \braces{ c' \in \Coin \mid (a', c') \in \newhc  }}\;.
	\end{align*}
	
	Then,
	$\newsum = \oldsum + 1$,
	$\newbal(a)= \oldbal(a)+1$. Moreover, for all other addresses $a' \neq a$, 
	it holds $\newbal(a')=\oldbal(a')$. \qed
\end{restatable}

\noindent
Smart transitions minting an arbitrary number of $n$ coins, as in our
~\Cref{fig:mint-n-high-level}, is then realized by 
repeating the $\mint_1$ transition $n$ times.
Based on the soundness of $\mint_1$, ensuring that $\mint_1$
preserves smart invariants,   
we conclude by induction that $n$ repetitions of $\mint_1$, that is
\emph{minting $n$ coins,
also preserves smart invariants. } The precise definition of $\mint_n$ together with the soundness result is stated in \Cref{apd:mint_n}.


\subsection{Completeness Relative to a Translation
  Function}\label{sec:enc_sound}
Smart invariants provide sufficient conditions for ensuring soundness
of our \SL{} encodings (\Cref{thm:enc:soundness}). We next show
that, under additional constraints, smart invariants are also necessary
conditions, establishing thus \emph{(relative) completeness of our
  encodings. }

A straightforward extension of \Cref{thm:enc:soundness} 
however does not hold. 
Namely, only  under the assumptions of
\Cref{thm:enc:soundness},
the following formula is not valid:
\[
\msum = \suma \bal(a) 
\quad \iff \quad 
\inv(\hc{},\;\ac{})\,.\] 

As a counterexample, assume 
(i) 
$\msum = \abs{\ac}$,
(ii) 
for every  $a\in\Addr$ it holds that 
$\bal(a) = \abs{ \braces{ c \in \Coin \mid (a,c) \in \hc } }$, that is
the assumptions of \Cref{thm:enc:soundness}. Further, let 
(iii)  
      the smart invariants \inv(\hc{},\;\ac{})
      hold for all but the coins $c_1,c_2\in\Coin$ 
and all but the addresses $a_1,a_2 \in \Addr$. We also assume 
%
that (iv) $c_1$ is active but not owned by any address and 
(v) $c_2$ is active and owned by the two distinct addresses 
$a_1,a_2$. We thus have 
$\msum = \suma \bal(a)$, yet $\inv(\hc{},\;\ac{})$ does not hold. 

To ensure completeness of our encodings, 
we therefore introduce a translation
function $f$ that restricts the set 
$\mF \define 2^{\Addr \times \Coin} \times 2^{\Coin}$  
of $(\hc,\ac)$ pairs, as follows. 
We exclude from $\mF$ those pairs 
$(\hc, \ac)$ that 
violate smart invariants by both 
(i) not satisfying (I2), 
as (I2) ensures that there are not too many active coins,
and by (ii) 
not satisfying at least one of (I1) and (I3), as (I1) and (I3) ensure that there are not too few active coins. 
The required translation function $f$ (\Cref{app:f}) now assigns every pair $(\bal, \msum)$ the set of all $(\hc,\ac) \in \mF$ that satisfy $\msum = \abs{ \ac }$, $\bal(a)= \abs{ \braces{ c \in \Coin \mid \hc(a,c) } }$ for every address $a$ and have not been excluded.
 





\begin{restatable}[Relative Completeness of \SL{} Encoding]{theorem}{completeImplicitSL}\label{sound_comp}
	Let 
	$(\bal, \msum) \in \N^{\Addr} \times \N$ 
	and let
	$(\hc,\ac) \in f(\bal,\msum)$ be arbitrary. 
	Then, 
	\[ 
	\msum = \suma \bal(a) 
	\quad \iff \quad  
	\inv\,(\hc{},\;\ac{})\;.\qquad\qed
	\]\vspace*{-5mm}
\end{restatable}

\subsection{\SL{} Encodings using Explicit Balances and Sums}\label{sec:expl_enc}

We now restrict our \SL{} encoding from Section~\ref{sec:enc_uf} to
explicitly reason with balance functions during smart transitions. 
We do so by expressing our translation function $f$ from Section~\ref{sec:enc_sound} in first-order logic.
We now use the  summation constant 
$\msum \in \N$ 
and the balance function  
$\bal : \Addr \to \N$ 
in our \SL{} encoding. In particular, we use our smart
  invariants 
$\inv(\hc{},\;\ac{})$  
in this explicit \SL{} encoding together with two  additional axioms
(Ax1, Ax2), ensuring that
$\msum = \abs{\ac}$ and 
$\bal(a) = \abs{ \braces{ c \in \Coin \mid \hc(a,c) } }$ 
for all $a\in\Addr$. 

To formalize the additional properties, we introduce two counting mechanisms in our \SL{}
encoding. 
The first one is a bijective function 
$\mcount: \Coin \to \N^+$ 
and the second one is a function 
$\ind: \Addr \times \Coin \to \N^+$, 
where $\ind(a,.): \Coin \to \N^+$ 
is bijective for every $a \in \Addr$.
To ensure that $\mcount$ and $\ind(a,.)$ count coins, we impose the
following two properties:
\[\tag{Ax1} 
	\forall c : \Coin.\; \ac(c) \iff \mcount(c) \leq \msum \;, 
\]
\[\tag{Ax2} 
	\forall c : \Coin.\; \forall a : \Addr.\; \hc(a,c) \iff \ind(a,c) \leq \bal(a) \;.  
\]

\Cref{fig:expl_enc} illustrates our revised \SL{} encoding for
our smart transition  $\mint_1$. We next ensure
soundness of our resulting explicit encoding for summation, as follows. 

\begin{restatable}[Soundness of Explicit \SL{} Encodings]{theorem}{soundExplicitSL}\label{fo_expressible}
 Let there be a pair 
$(\bal,\msum) \in \N^{\Addr} \times \N$, 
a pair 
$(\hc,\ac) \in \mF$,
and functions
$\mcount: \Coin{} \to \N^+$ 
and
$\ind: \Addr \times \Coin \to \N^+$.

Given that
$\mcount$ is bijective,
$\ind(a,.): \Coin \to \N^+$ 
is bijective for every $a\in\Addr$,
and that (Ax1), (Ax2) and $\inv\,(\hc{},\;\ac{})$ hold,
then, 
$\msum = \abs{\ac}$ and 
$\bal(a) = \abs{\braces{ c \in \Coin{}: \hc(a,c) } }$, for
every $a\in\Addr$.

In particular, we have $\msum = \suma \bal(a)$.\qed
\end{restatable}

\begin{figure}[tb]
	\centering
	\includegraphics[width=0.9\textwidth]{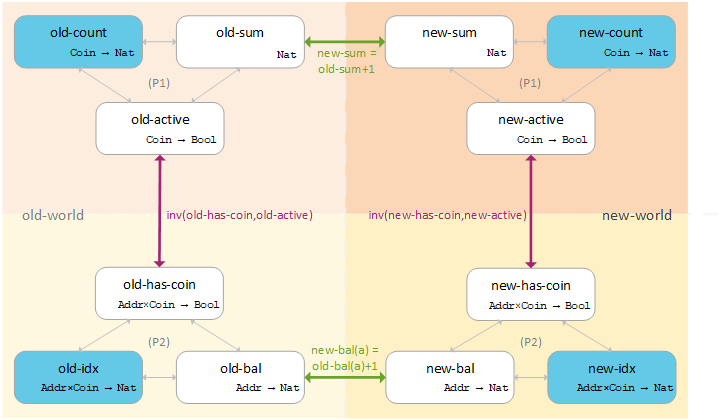}
	\caption{Explicit \SL{} Encoding of $\mint_1$, where \texttt{Addr} is short for \Addr.}
	\label{fig:expl_enc} \vspace{-0.3cm}
\end{figure}


When compared to \Cref{sec:enc_uf},
  our explicit \SL{} encoding introduced above uses
  our smart invariants as axioms of our encoding, together with (Ax1)
  and (Ax2). 
In our explicit \SL{} encoding, the post-conditions asserting functional correctness of smart
transitions express thus relations among  $\oldsum$ to $\newsum$.
For example, for $\mint_n$ we are
interested in ensuring \vspace{-0.2cm}
\begin{equation}\label{eq:mintN:nototal} \mint_n \To \newsum=\oldsum +n \;. 
\end{equation}
 By using two new constants
$\texttt{old-total}$, $\texttt{new-total} \in \N $, we can use
$\msum=\texttt{total}$ as smart invariant for $\mint_n$. As a result, the property to be ensured is then
\begin{equation}\label{eq:expl:mintN}
\begin{array}{l}
(
\oldsum=\texttt{old-total} 
\land 
\texttt{new-total}=\texttt{old-total}+n 
\land 
\mint_n
)\\
\To 
\parens{
	\newsum=\texttt{new-total}
  }\;.
  \end{array}
\end{equation}

It is easy to see that the negations of \eqref{eq:mintN:nototal} and \eqref{eq:expl:mintN}  are equisatisfiable. 
We note however that the additional constants $\texttt{old-total}$,
$\texttt{new-total}$ used in~\eqref{eq:expl:mintN}  lead to unstable
results within automated reasoners, as discussed in Section \ref{sec:exp}.

\section{Experiments}\label{sec:exp}

\noindent{\bf From Theory to Practice.}
To make our explicit \SL{} encodings handier for automated reasoners, we improved the setting illustrated in \Cref{fig:expl_enc} 
by applying the following restrictions without losing any generality. 

\noindent (i) 
	The predicates {\hc{} and \ac{} were removed from the explicit \SL{}
          encodings}, by 
	replacing them by their equivalent expressions (Ax1)-(Ax2).
	
\noindent (ii)  
	The {surjectivity assertions of \mcount{} and \ind{} were restricted} 
	to the relevant intervals $[1, \msum]$, $[1, \bal(a)]$ respectively.
	
\noindent (iii) 
	Compared to \Cref{fig:expl_enc}, 
	only {one mutual \mcount{} and one mutual \ind{}}
	functions were used. We however conclude that we do not lose
       expressivity of our resulting \SL{} encoding, as
          shown in \Cref{app:onecount}. 
	
\noindent (iv) 
	When our \SL{} encoding contains expressions such as
	$\forall c : \Coin.\; \ind(a_0,c) \in [l_0,u_0] \iff \ind(a_1,c) \in [l_1,u_1]$, 
	with $a_0$, $a_1$ being distinct addresses such that  either 
	$u_i \leq \bal(a_i)$ 
	or 
	$l_i > \bal(a_i)$, 
	$i\in \braces{ 0,1 }$, 
	then it can be assumed that the {coins in those intervals are in the same order} 
	for both functions (\Cref{app:ordering}).

Based on the above,
we derive three different explicit \SL{}
encodings to be used in automated reasoning about smart
transitions. We respectively denote these explicit \SL{} encodings by 
\intenc{}, \natenc{} and \idenc{}, and describe them next.

\noindent{\bf Benchmarks.} 
In our experiments, 
we consider four smart transitions $\mint_1$, $\mint_n$, $\transfer_1$
and $\transfer_n$, respectively denoting minting and transferring
one and $n$ coins. 
These transitions capture the main operations of linear integer arithmetic.
In particular, $\mint_n$ implements the smart transition of  our
running example from \Cref{fig:mint-n-high-level}.

For each of the four smart transitions, we implement four \SL{} encodings: 
the implicit \SL{} encoding  {\tt uf} from \Cref{sec:enc_uf} using only uninterpreted functions  
and three explicit encodings \intenc{}, \natenc{} and \idenc{}
as variants of  \Cref{sec:expl_enc}.
We also consider three additional arithmetic benchmarks using
\intenc{}, which are not directly motivated by smart contracts. 
 Together with variants of \intenc{} and \natenc{} presented in the sequel, our benchmark set contains  31 examples altogether,
  with each example being formalized in the SMT-LIB input
  syntax~\cite{smtlib}. In addition to our encodings, we also proved
  consistency of the axioms used in our encodings.


\noindent{\bf \SL{} Encodings and Relaxations.} 
Our explicit \SL{} encoding \intenc{} uses linear integer arithmetic, 
whereas \natenc{} and \idenc{} are based on natural numbers. 
As naturals 
are not a built-in theory in SMT-LIB, 
we assert the axioms of Presburger arithmetic directly in the
encodings of \natenc{} and \idenc{}. 

\begin{figure}[tb]
	\centering
	\vspace{-0.2cm}
	\includegraphics[width=0.48\textwidth]{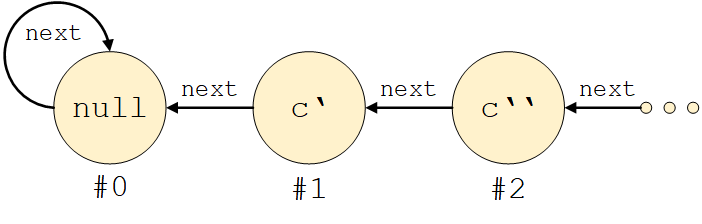}
	\caption{Linked Lists in \texttt{id}.}
	\label{fig:id_pic}
\end{figure}

In our \idenc{} encodings, inductive datatypes are additionally used to order coins. 
There exists one linked list of all coins for \mcount{} 
and one for each $\ind(a,.)$, $a \in \Addr$. 
Additionally, there exists a ``null'' coin, 
which is the first element of every list and is not owned by any
address.  
As shown in \Cref{fig:id_pic}, the numbering of each coin is defined by 
its position in the respective list. 
This way surjectivity for \mcount{} and \ind{} can respectively be asserted by the formulas
$\exists c : \Coin. \;\mcount(c) \approx \msum$
and
$\forall a : \Addr. \; \exists c : \Coin. \;\ind(a,c) \approx \bal(a)$.
However, asserting surjectivity for \intenc{} and \natenc{} cannot be
achieved without quantifying over $\N^+$.
Such quantification would drastically effect the performance of 
automated reasoners in (fragments of) first-order logics. As a remedy,
within the default encodings of \intenc{} and \natenc{}, we only consider relevant instances of surjectivity. 

Further,  we consider  variations of \intenc{} and \natenc{} by
asserting proper surjectivity to the relevant intervals of  \ind{} and
\mcount{} (denoted as \emph{surj}) and/or adding the \total{}
constants mentioned in Section \ref{sec:expl_enc} (denoted as
\emph{with \total{}, no \total}) .
These variations of \intenc{} and \natenc{} are implemented for $\mint_1$ and $\transfer_1$.

\noindent{\bf Experimental Setting.} 
We evaluated our benchmark set of 31 examples using SMT solvers Z3~\cite{Z3} and CVC4~\cite{CVC4},
as well as the first-order theorem prover Vampire~\cite{Vampire}. Our
experiments were run on a standard machine with an Intel Core i5-6200U CPU (2.30GHz, 2.40GHz) and 8 GB RAM. The time is given in seconds and we ran all experiments with a time limit of 300s. Time out is indicated by the symbol $\times$. The default parameters were used for each solver, unless stated otherwise in the corresponding tables. The precise calls of the solvers, 
together with examples of the encodings, 
can be found in \Cref{app:codes}\footnote{All encodings are available at github.com/SoRaTu/SmartSums.}.

\begin{table}[tb]
	\centering
	\begin{tabular}{|l|rrr|l|rrr|}
	\hline
		\multicolumn{4}{|c|}{$\mint_1$} &\multicolumn{4}{c|}{$\transfer_1$} \\
	\hline
		no \total{}
			& \multicolumn{1}{l}{\bf $\,$ Z3 \qquad} 
			& \multicolumn{1}{l}{\bf $\;$CVC4\quad} 
			& \multicolumn{1}{l|}{\bf $\,$Vampire \quad} 
			&  no \total{}
			& \multicolumn{1}{l}{\bf $\,$ Z3 \qquad} 
			& \multicolumn{1}{l}{\bf $\;$CVC4 \quad} 
			& \multicolumn{1}{l|}{\bf $\,$Vampire \quad} 
			\\
			
		$\;$ \natenc{}
			& \qquad 0.02 
			& $\times$ 
			& 0.92 
			
			& $\;$ \natenc{}
			& $\;$ $\times$   
			& $\times$  
			& 15.35 
			\\
			
		$\;$ \natenc{} surj. 
			& $\;$ $\times$   
			& $\;$ $\times$  
			& $\times$    
			
			& $\;$ \natenc{} surj. 
			& \quad 100.03  
			& $\;$ $\times$  
			& $\times$  
			\\
			
		$\;$ \intenc{}
			& 0.02 
			& 0.03 
			& $\times$  
			
			& $\;$ \intenc{}
			& 0.02 
			& 0.07 
			& $\times$   
			\\
			
		$\;$ \intenc{} surj. \quad 
			& $\times$  
			& \qquad 5.96 
			& $\times$  
			
			& $\;$ \intenc{} surj. \quad 
			& 1.02 
			& $\times$ 
			& $\times$   
			\\

	\hline
		 with \total{}$\,$ 
			& \multicolumn{1}{l}{\bf $\,$ Z3} 
			& \multicolumn{1}{l}{\bf $\;$CVC4\quad}
			&\multicolumn{1}{l|}{\bf $\,$Vampire \quad} 
			
			&  with \total{}$\,$ 
			& \multicolumn{1}{l}{\bf $\,$ Z3} 
			& \multicolumn{1}{l}{\bf $\;$CVC4\quad} 
			&\multicolumn{1}{l|}{\bf $\,$Vampire \quad} 
		\\
		
		$\;$ \natenc{} 
			& 0.03 
			& $\times$  
			& 2.92 
			
			& $\;$ \natenc{}
			& 0.28 
			& $\times$  
			& 22.54 
			\\
			
		$\;$ \natenc{} surj. 
			& 0.11 
			& $\times$  
			& $\times$  
			
			& $\;$ \natenc{} surj. 
			& 38.24 
			& $\times$ 
			& $\times$  
			\\
			
		$\;$ \intenc{} 
			& 0.02 
			& 0.03 
			& $\times$  
			
			& $\;$ \intenc{} 
			& 0.02 
			& 0.10 
			& $\times$  
			\\
			
		$\;$ \intenc{} surj. 
			& 3.81 
			& 5.95 
			& $\times$  
			
			& $\;$ \intenc{} surj. 
			& $\times$  
			& \qquad 6.56 
			& $\times$  
			\\
	\hline
	\end{tabular}
	\vspace{0.2cm}
	\caption{
		Results of $\mint_1$ and 
		 $\texttt{transferFrom}_1$ using \natenc{} and \intenc{},
		with/without the \total{} Constants and  
		with/without Surjectivity. 
	}
	\label{tbl:total_surj_table}	\vspace{-7mm}
      \end{table}

\noindent{\bf Experimental Analysis.} We first report on our 
experiments using different variations of \intenc{} and \natenc.
\Cref{tbl:total_surj_table} shows that asserting complete surjectivity for \intenc{} and \natenc{} is computationally hard and indeed
significantly effects the performance of automated reasoners. 
Thus, for the following experiments 
only relevant instances of surjectivity, 
such as 
$\exists c : \Coin.\, \mcount(c) = \msum$
were asserted in \intenc{} and \natenc{}.
\Cref{tbl:total_surj_table} also illustrates the instability of
using the 
$\texttt{total}$ constant. 
Some tasks seem to be easier even though 
their reasoning difficulty increased strictly by adding additional constants. 
      
      \begin{table}[tb]
      	\centering
      	\begin{tabular}{|c|lr|lr|lr|lr|}
      		\hline
      		\multicolumn{1}{|c|}{\multirow{2}{*}{Encoding}} & \multicolumn{8}{c|}{Task}  \\
      		\cline{2-9}
      		& \multicolumn{2}{c|}{$\mint_1$} &\multicolumn{2}{c|}{$\transfer_1$}& \multicolumn{2}{c|}{$\mint_n$} & \multicolumn{2}{c|}{$\transfer_n$}\\
      		\hline
      		\texttt{uf}  &
      		\begin{tabular}{l}
      			Z3:\\
      			CVC4:\\
      			Vampire: \\
      		\end{tabular} & 	\begin{tabular}{r}
      			0.01\\
      			0.02\\
      			0.18\\
      		\end{tabular} &
      		\begin{tabular}{l}
      			Z3:\\
      			CVC4:\\
      			Vampire: \\
      		\end{tabular} & 	\begin{tabular}{r}
      			0.02\\
      			0.03\\
      			0.19\\
      		\end{tabular} &
      		\begin{tabular}{l}
      			Z3:\\
      			CVC4:\\
      			Vampire: \\
      		\end{tabular} & 	\begin{tabular}{r}
      			$\times$\\
      			$\times$\\
      			$0.35^{*}$\\
      		\end{tabular} &
      		\begin{tabular}{l}
      			Z3:\\
      			CVC4:\\
      			Vampire: \\
      		\end{tabular} & 	\begin{tabular}{r}
      			$\times$\\
      			$\times$\\
      			$0.44^{*}$\\
      		\end{tabular} \\
      		\hline
      		\texttt{nat} &
      		\begin{tabular}{l}
      			Z3:\\
      			CVC4:\\
      			Vampire: \\
      		\end{tabular} & 	\begin{tabular}{r}
      			0.02\\
      			$\times$\\
      			0.92\\
      		\end{tabular} &
      		\begin{tabular}{l}
      			Z3:\\
      			CVC4:\\
      			Vampire: \\
      		\end{tabular} & 	\begin{tabular}{r}
      			$\times$\\
      			$\times$\\
      			15.35\\
      		\end{tabular} &
      		\begin{tabular}{l}
      			Z3:\\
      			CVC4:\\
      			Vampire: \\
      		\end{tabular} & 	\begin{tabular}{r}
      			$\times$\\
      			$\times$\\
      			$23.23^{\dagger}$\\
      		\end{tabular} &
      		\begin{tabular}{l}
      			Z3:\\
      			CVC4:\\
      			Vampire: \\
      		\end{tabular} & 	\begin{tabular}{r}
      			$\times$\\
      			$\times$\\
      			$228.22^{\dagger}$ \\
      		\end{tabular} \\
      		\hline
      		\texttt{int}  &
      		\begin{tabular}{l}
      			Z3:\\
      			CVC4:\\
      			Vampire: \\
      		\end{tabular} & 	\begin{tabular}{r}
      			0.02\\
      			0.03\\
      			$\times$ \\
      		\end{tabular} &
      		\begin{tabular}{l}
      			Z3:\\
      			CVC4:\\
      			Vampire: \\
      		\end{tabular} & 	\begin{tabular}{r}
      			0.02\\
      			0.07\\
      			$\times$ \\
      		\end{tabular} &
      		\begin{tabular}{l}
      			Z3:\\
      			CVC4:\\
      			Vampire: \\
      		\end{tabular} & 	\begin{tabular}{r}
      			0.03\\
      			0.05\\
      			$\times$ \\
      		\end{tabular} &
      		\begin{tabular}{l}
      			Z3:\\
      			CVC4:\\
      			Vampire: \\
      		\end{tabular} & 	\begin{tabular}{r}
      			0.11\\
      			0.35\\
      			$\times$ \\
      		\end{tabular} \\
      		\hline
      		\texttt{id} &
      		\begin{tabular}{l}
      			Z3:\\
      			CVC4:\\
      			Vampire: \\
      		\end{tabular} & 	\begin{tabular}{r}
      			$\times$\\
      			$\times$\\
      			$7.36^{\ddagger}$ \\
      		\end{tabular} &
      		\begin{tabular}{l}
      			Z3:\\
      			CVC4:\\
      			Vampire: \\
      		\end{tabular} & 	\begin{tabular}{r}
      			$\times$\\
      			$\times$\\
      			$17.16^{\ddagger}$ \\
      		\end{tabular} &
      		\begin{tabular}{l}
      			Z3:\\
      			CVC4:\\
      			Vampire: \\
      		\end{tabular} & 	\begin{tabular}{r}
      			$\times$\\
      			$\times$\\
      			$23.52^{\ddagger}$ \\
      		\end{tabular} &
      		\begin{tabular}{l}
      			Z3:\\
      			CVC4:\\
      			Vampire: \\
      		\end{tabular} & 	\begin{tabular}{r}
      			$\times$\\
      			$\times$\\
      			$\times$ \\
      		\end{tabular} \\
      		\hline
      	\end{tabular}
      	\vspace{0.2cm}
      	\caption{
      		Smart Transitions using Implicit/Explicit \SL{} Encodings. 
      	}
      	\label{tbl:overview_table} \vspace{-0.8cm}
      \end{table}

Our most important experimental findings are shown in
\Cref{tbl:overview_table}, demonstrating that \emph{our \SL{} encodings are
suitable for automated reasoners. 
Thanks to our explicit \SL{} encodings, each solver can certify every
smart transition in at least one  encoding.} 
Our explicit \SL{} encodings are more relevant than the implicit encoding \ufenc{} as we can express and compare any two non-negative integer sums, 
whereas for \ufenc{} handling arbitrary values $n$ 
can only be done by iterating over the $\mint_1$ (or $\transfer_1$)
transition. 
This iteration requires inductive reasoning, 
which currently only Vampire could do~\cite{inductive-reasoning},
as indicated by the superscript $*$. 
Nevertheless, the transactions $\mint_1$, $\transfer_1$, 
which involve only one coin in \ufenc{}, require no inductive
reasoning 
as the actual sum is not considered; each of our solvers can certify
these examples.

We note that the tasks $\mint_n$ and $\transfer_n$ from
\Cref{tbl:overview_table}  yield a huge search space when using their
explicit \SL{} encodings within automated reasoners. 
We split these tasks into proving intermediate lemmas and proved each
of these lemmas independently, by the respective solver. In
particular, 
we used one lemma for $\mint_n$ and four lemmas for $\transfer_n$.
In our experiments, we only used the recent theory reasoning framework
of Vampire with split queues~\cite{splitqueues} and indicate our
results in by superscript $\dagger$. 

We further remark that our explicit \SL{} encoding \idenc{} using
inductive datatypes also requires inductive reasoning about smart
transitions and beyond.  
The need of induction explains why SMT solvers failed proving our
\idenc{} benchmarks, as shown in \Cref{tbl:overview_table}. We
note that 
Vampire  found a proof using built-in induction~\cite{inductive-reasoning} and
theory-specific reasoning~\cite{splitqueues}, as indicated by
superscript $\ddagger$. 

We conclude by showing the generality of our approach beyond smart
transitions. It in fact enables fully automated reasoning about any two summations $\sum_{i \in I} g(i)$, $\sum_{i \in I} h(i)$ of non-negative integer values $g(i)$, $h(i)$ ($i \in I$) over a mutual finite set $I$. The examples of \Cref{tbl:arith_table} affirm this claim.

\begin{table}[tb]
	\centering
	\begin{tabular}{|l|l|lr|}
		\hline
		\multicolumn{2}{|c|}{Task} &\multicolumn{2}{c|}{\multirow{2}{*}{Time}} \\
		\cline{1-2}
		\multicolumn{1}{|c|}{Transition} & \multicolumn{1}{c|}{Impact} &&\\
		\hline
		\begin{tabular}{l}
			$\newbal(a_0)=\oldbal(a_0)+3$\\
			$\newbal(a_1)=\oldbal(a_1)-3$\\
		\end{tabular} & $\newsum=\oldsum$ & 
			\begin{tabular}{l}
			Z3:\\
			CVC4:\\
			Vampire: \\
		\end{tabular} & 	\begin{tabular}{r}
		0.20\\
		1.28\\
		$\times$ \\
	\end{tabular}
	\\
		\hline
		\begin{tabular}{l}
			$\newbal(a_0)=\oldbal(a_0)+4$\\
			$\newbal(a_1)=\oldbal(a_1)-2$\\
		\end{tabular} & $\newsum=\oldsum+2$ & 		\begin{tabular}{l}
		Z3:\\
		CVC4:\\
		Vampire: \\
	\end{tabular}		&\begin{tabular}{r}
		0.58\\
		7.14\\
		$\times$ \\
	\end{tabular}\\
		\hline	
		\begin{tabular}{l}
			$\newbal(a_0)=\oldbal(a_0)+5$\\
			$\newbal(a_1)=\oldbal(a_1)-3$\\
			$\newbal(a_2)=\oldbal(a_2)-1$\\
		\end{tabular} & $\newsum=\oldsum +1$ & 		\begin{tabular}{l}
		Z3:\\
		CVC4:\\
		Vampire: \\
	\end{tabular}	&	\begin{tabular}{r}
		1.52\\
		155.20\\
		$\times$ \\
	\end{tabular}\\
		\hline
	\end{tabular}
\vspace{0.2cm}
	\caption{Arithmetic Reasoning in the Explicit \SL{} Encoding \texttt{int}.} 
	\label{tbl:arith_table} \vspace{-0.8cm}
\end{table}

\section{Related work}

\noindent{\it Smart Contract Safety.}
Formal verification of smart contracts is an emerging hot topic because of 
the value of the assets stored in smart contracts, 
e.g. the DeFi software~\cite{DeFiPulse}. 
Due to the nature of the blockchain,
bugs in smart contracts are irreversible and thus the demand for
provably bug-free smart contracts is high. 

The K interactive framework has been used to verify safety of a smart contract, 
e.g. in~\cite{CAV20Deposit}.
Isabelle~\cite{isabelle} was also shown to be useful in manual, interactive verification
of smart contracts~\cite{isabelle-smart-contracts}.
We, however, focus on automated approaches.

There are also efforts to perform deductive verification of smart contracts both on 
the source level in languages such as Solidity~\cite{azureblockchain,EF,SRI} and Move~\cite{CAVMove},
as well as on the the Ethereum virtual machine (EVM) level~\cite{Certora,ccs/SchneidewindGSM20}.
This paper improves the effectiveness of these approaches by 
developing techniques for automatically reasoning about unbounded
sums. 
This way, we believe we support a more semantic-based verification of smart contracts.

Our approach differs from works using ghost variables \cite{SRI}, 
since we do not manually update the ``ghost state''.
Instead, the verifier needs only to reason about the local changes,
and the aggregate state is maintained by the axioms. 
That means other approaches assume (a) the local changes and (b) the impact on ghost variables (sum), 
whereas we only assume (a) and automatically prove $a \To b$. 
This way, we reduce the user-guidance in providing and proving (b).	

Our work complements approaches that verify smart contracts 
as finite state machines~\cite{azureblockchain} and methods, like
ZEUS~\cite{zeus}, using symbolic model checking
and abstract interpretation to verify generic safety properties for smart contracts.

The work in \cite{shuvendu} provides an extensive evaluation of ERC-20 and ERC-721 tokens. 
ERC-721 extends ERC-20 with ownership functions, 
one of which being ``approve''. 
It enables transactions on another party’s behalf. 
This is independent of our ability to express sums in first-order logic, 
as the transaction’s initiator is irrelevant to its effect.

\noindent{\it Reasoning about Financial Applications.}
Recently, the Imandra prover introduced an automated reasoning framework 
for financial applications~\cite{imandra-reasoning,imandra,financial}.
Similarly to our approach,
these works use SMT procedures to verify and/or generate counter-examples to 
safety properties of low- and high-level algorithms.
In particular, results of~\cite{imandra-reasoning,imandra,financial}
include examples of verifying ranking orders in matching logics of exchanges,
proving high-level properties such as transitivity and anti-symmetry
of such orders. 
In contrast, 
we focus on verifying properties relating 
local changes in balances to changes of the global state 
(the sum). 
Moreover, our encodings enable automated reasoning both in 
SMT solving and first-order theorem proving.

%
%

\noindent{\it Automated Aggregate Reasoning.}
The theory of first-order logic with aggregate operators
has been thoroughly studied in~\cite{logic-with-aggregates,fmt}.
Though proven to be strictly more expressive than first-order logic,
both in the case of general aggregates as well as simple counting logics,
in this paper we present a practical way to encode a weakened version of aggregates 
(specifically sums) 
in first-order logic.
Our encoding (as in \Cref{sec:encoding}) works by expressing
particular sums of interest, 
harnessing domain knowledge to avoid the need of general aggregate operators.

Previous works~\cite{bapa,extending-smt-to-hol} 
in the field of higher-order reasoning do not directly
discuss aggregates.
The work of~\cite{bapa}  extends Presburger arithmetic with 
Boolean algebra for finite, unbounded sets of uninterpreted elements.
This includes a way to express the set cardinalities  and to compare
them against integer variables, 
but does not support uninterpreted functions,
such as the balance functions we use throughout our approach. 

The SMT-based framework of~\cite{extending-smt-to-hol} takes a different, 
white-box approach,
modifying the inner workings of SMT solvers to support higher-order logic.
We on the other hand treat theorem provers and SMT solvers as black-boxes,
constructing first-order formulas that are tailored to their capabilities.
This allows us to use any off-the-shelf SMT solver.

In~\cite{dpll-agg}, an SMT module for the theory of FO(Agg) is presented,
which can be used in all DPLL-based SAT, SMT and ASP solvers.
However, FO(Agg) only provides a way to express functions that have
sets or similar constructs as inputs, but not to verify their semantic behavior.

\section{Conclusions}\label{sec:conclusions}
We present a methodology for reasoning about unbounded sums
in the context of \emph{smart transitions}, that is 
transitions that occur in smart contracts modeling transactions.
%
Our sum logic \SL{} and its usage of sum constants,
instead of fully-fledged sum operators,
turns out to be most appropriate for the setting of smart contracts. We
show that \SL{} has decidable fragments (\Cref{sec:decidable-sum-logic}),
as well as undecidable ones (\Cref{sec:undecidable}). 
Using two phases to first implicitly encode \SL{} in first-order logic
(\Cref{sec:enc_uf}),
and then explicitly encode it (\Cref{sec:expl_enc}),
allows us to use off-the-shelf automated reasoners in new ways,
and automatically verify the semantic correctness of smart transitions.


Showing the (un)decidability of the \SL{} fragment with 
two sets of uninterpreted functions and sums is an interesting step for
further work, as this fragment supports encoding smart transition systems. 
Another interesting direction of future work is to apply our approach to different aggregates, 
such as minimum and maximum  
and to reason about under which conditions these values stay above/below certain thresholds. 
A slightly modified setting of our \SL{} axioms can already handle 
$\min$/$\max$ aggregates in a basic way, 
namely by using $\geq$ and $\leq$ instead of equality and 
dropping the injectivity/surjectivity (respectively) axioms of the counting mechanisms.

Summing upon multidimensional arrays in various ways is yet another direction of future research. 
Our approach supports the summation over all values in all dimensions by 
adding the required number of parameters to the predicate \ind{} and by adapting the axioms accordingly.

\section*{Acknowledgements}
We thank Petra Hozzov\'a for fruitful  discussions on our encodings and Sharon Shoham-Buchbinder for her insights and contributions to this paper. This  work was partially 
funded by the ERC CoG ARTIST 101002685, the ERC StG SYMCAR 639270,
the United States-Israel Binational Science Foundation (BSF) grant No. 2016260, Grant No. 1810/18 from the Israeli Science Foundation, Len Blavatnik and the Blavatnik Family foundation,
the FWF grant LogiCS W1255-N23, the TU Wien DK SecInt and the Amazon ARA 2020 award FOREST.

\newpage

\bibliographystyle{llncs2e/splncs04}
\bibliography{includes/references}

\appendix

\newpage
\section{Proofs for Sum Logic}\label{app:sum-logic-proofs}
\begin{definition}[\SL{} Structure]
	\label{def:sum-structure}
	
	Let $\Sigma$ be an \SL{} vocabulary. 
	We write a structure $\A = (\domain, \interp) \in \structure{\Sigma}$ as a tuple
	\[
	\A = \parens{ 
		A,
		\N,
		a^\A_1, \dots, a^\A_l,
		b^\A_1, \dots, b^\A_m,
		c^\A_1, \dots, c^\A_d,
		s^\A_1, \dots, s^\A_m,
		0, 1, +, \leq
	}
	\]
	where $A = \domain(\Addr)$ is some \emph{finite}
	\footnote{
		In fact, we need only to require that the set of addresses with non-zero balances 
		$\braces{ \alpha \in \domain(\Addr) \mid \forall j. b^\A_j(\alpha) > 0 }$ be finite. 
		Except for addresses that are referred by an \Addr{} constant, 
		we can always discard all zero-balance addresses from a model. 
		Thus, we might as well limit ourselves to finite sets of addresses.
	}, possibly empty set. We have 
	$a_i^\A = \interp(a_i) \in A$; 
	$b^\A_j = \interp(b^1_j) \in \N^A$; and 
	$c_k^\A = \interp(c_k), s^\A_j = \interp(s_j) \in \N$. 
	
	We always assume that $\domain(\Nat) = \N$, 
	and that $0,1,+^2$ and $\leq^2$ are interpreted naturally. 
	For brevity, we omit them when describing \SL{} structures.
\end{definition}

\subsection*{Distinct Models Proof}
\begin{observation}
	For any set $X$ and any partition $P$ thereof, 
	it holds that $\abs{P} \leq \abs{X}$.
\end{observation}

\begin{definition}[Partition-Induced Function]
	\label{def:partition-induced-function}
	
	Let $P$ be a partition of a finite set $X$ of size $l$. 
	$P = \braces{A_1, \dots, A_{l'}}$ where $l' \leq l$.		
	
	We define the partition-induced function $f_P(x)$ (for any $x \in X$) 
	as the index $i$ such that $A_i \in P$ and $x \in A_i$.
	
	For brevity, we denote $f_P(x)$ as $P(x)$.
\end{definition}

\begin{definition}[Function-Induced Equivalence Class]
	\label{def:function-induced-class}
	
	Let $f$ be some function over some set $X$. 
	We define the function-induced equivalence class for each $x \in X$ as
	\[
	[x]_f \define \braces{ x' \in X \mid f(x') = f(x) }.
	\]
\end{definition}

\begin{definition}[Function-Induced Partition]
	\label{def:function-induced-partition}
	
	Let $f$ be some function defined over some set $X$. 
	We define the function-induced partition $P_f$ as
	\[
	P_f \define \braces{[x]_f \mid x \in X }.
	\]
\end{definition}

\begin{definition}[Partitioning Sum Terms by $P$]
	\label{def:partitioning-terms}
	
	Let $t$ be some term over an \SL{} vocabulary $\Sigma = \Sigma^{l,m,d}$ 
	(with $l$ \Addr{} constants) 
	and let $P$ be some partition of $\braces{a_1,\dots,a_l}$.
	
	We define a transformation $\mathcal{T}_P(t)$ inductively as a term over an \SL{} vocabulary 
	$\Sigma_P = \Sigma^{l',m,d}$ 
	with 
	$l' = \abs{P} \leq l$ \Addr{} constants:	
	\[
	\mathcal{T}_P(t) 
	\define \left\{
	\begin{array}{ll}
		a_{P(a_i)}
		& \mbox{if } t = a_i \\
		x_i
		& \mbox{if } t = x_i \text{ of sort } \Addr\\
		s_j
		& \mbox{if } t = s_j \\
		b_j(\mathcal{T}_P(t_1)) 
		& \mbox{if } t = b_j(t_1) 
		\text{ where $t_1$ is some $a_i$ or $x_i$} \\
		\mathcal{T}_P(t_1) + \mathcal{T}_P(t_2) 
		& \mbox{if } t = t_1 + t_2
	\end{array}
	\right.
	\]
\end{definition}

\begin{definition}[Partitioning an \SL{} Formula by $P$]
	\label{def:partitions-formulas}
	We naturally extend the terms transformation $\mathcal{T}_P$ to formulas.
\end{definition}

\begin{observation}
	For any \SL{} vocabulary $\Sigma$, $\Sigma_P \subseteq \Sigma$, 
	since $l' \leq l$. 
	Therefore, for any formula $\varphi$ in some fragment FRAG of \SL{}, 
	$\mathcal{T}_P(\varphi) \in \text{FRAG}$ as well.
\end{observation}

\begin{definition}[Distinct Structures]
	\label{def:distinct-structures}
	
	An \SL{} structure $\A$ is considered \emph{distinct} when
	$
	\abs{ \braces{ a_1^\A, \dots, a_l^\A } } = l
	$. 
	I.e. the $l$ \Addr{} constants represent 
	$l$ distinct elements in $\domain(\Addr)$.
\end{definition}

\distinctModelsTheorem*
\begin{longproof}{\Cref{thm:distinct-models}}
	~\\
	\proofpart{
		If $\varphi$ has an \SL{} model, 
		then there exists some partition $P$ such that 
		$\mathcal{T}_P(\varphi)$ 
		has a \emph{distinct} \SL{} model ($\To$)
	}
	
	Let $\A$ be some \SL{} model of $\varphi$ and let 
	$f$ be the mapping from $\braces{a_1,\dots,a_l}$ to $A$, i.e
	\[
	f(a_i) \triangleq a_i^\A
	\]
	
	Let $P$ be the partition (of size $l'$) induced by $f$ 
	and we construct a distinct \SL{} model 
	\[A' = \parens{
		A,
		a'_1,\dots,a'_{l'},
		b_1^\A, \dots, b_m^\A,
		c_1^\A,\dots,c_n^\A,
		s_1^\A, \dots, s_m^\A
	}
	\] 
	for $\mathcal{T}_P(\varphi)$, where 
	$A$, 
	$b_1^\A, \dots, b_m^\A$, 
	$c_1^\A,\dots,c_n^\A$, 
	and $s_1^\A, \dots, s_m^\A$ are taken from $\A$.
	
	For every $i' \in [1,l']$, $a'_{i'}$ is defined as 
	$
	a'_{i'} = a^\A_i
	$
	for some $i \in [1,l]$ such that $P(a_i) = i'$.
	
	\begin{remark}
		The choice of $i$ is unimportant, 
		since for any two indices $i_1,i_2$, 
		if $P(a_{i_1}) = i' = P(a_{i_2})$ 
		then by definition of $P$, $a^\A_{i_1} = a^\A_{i_2}$.
	\end{remark}
	
	\begin{observation}
		$\A'$ is distinct and holds the sum property.
	\end{observation}
	
	\begin{nclaim}\label{claim:equal-closed-terms-in-A-A'}
		For any closed term $t$ over $\Sigma$, 
		$\interp(t) = \interp'(\mathcal{T}_P(t))$ 
		(i.e. the interpretation of $t$ in $\A$ equals 
		to the interpretation of $\mathcal{T}_P(t)$ in $\A'$).
	\end{nclaim}	
	\begin{proof}
		Since $\mathcal{T}_P(t) = t$ for all terms except terms containing $a_i$, 
		and since $\A'$ is identical to $\A$ except for \Addr{} constants, 
		we only need to consider this kind of terms. 
		
		Moreover, since $\mathcal{T}_P$ is defined inductively, 
		it suffices to prove the claim for the basis terms $a_i$.
		
		Let $t = a_i$ for some $i \in [1,l]$, and let $i' = P(a_i)$:
		
		\begin{align*}
			\interp'(\mathcal{T}_P(t)) 
			&= \interp' ( \mathcal{T}_P(a_i) ) \\
			&= \interp' ( a_{i'} ) \\
			&= a'_{i'} \\
			&= a^\A_i \\
			&= \interp(a_i) = \interp(t)
		\end{align*}
	\end{proof}
	
	\begin{nclaim}
		\label{claim:equal-open-terms-in-A-A'}
		For any term $t$ with free variables 
		$x_1,\dots, x_r$ (of sort \Addr{}), 
		for all $\alpha_1,\dots,\alpha_r \in A$, 
		for any assignment $\Delta$, 
		let $\Delta' = \Delta \subs$, and therefore
		$
		\interp_{\Delta'}(t)
		=
		\interp'_{\Delta'}(\mathcal{T}_P(t))
		$.
	\end{nclaim}
	\begin{proof}
		Identical to the proof of 
		\Cref{claim:equal-closed-terms-in-A-A'}.
	\end{proof}
	
	\begin{nclaim}\label{claim:distinct-sat-1}
		Let $\xi$ be a sub-formula of $\varphi$, therefore:
		
		\begin{enumerate}
			\item 
				If $\xi$ is a closed formula then 
				$\A \sat \xi \iff \A' \sat \mathcal{T}_P(\xi)$
			
			\item 
				If $\xi$ is a formula with free variables 
				$x_1,\dots,x_r$ 
				then for any 
				$\alpha_1,\dots,\alpha_r \in A$,
				$
				\A \sat \xi \subs
				\iff
				\A' \sat \mathcal{T}_P(\xi) \subs
				$
		\end{enumerate}
	\end{nclaim}
	
	\noindent
	Since $\varphi$ is a closed formula, 
	and since $\A \sat \varphi$, 
	it holds that $\A' \sat \mathcal{T}_P(\varphi)$ 
	and therefore $\A'$ is a distinct \SL{} model for $\mathcal{T}_P(\varphi)$.
	
	\begin{proof}[\Cref{claim:distinct-sat-1}]
		Let us consider the following cases:
		
		\proofcase{$\xi = t_1 \approx t_2$ without free variables}
		Follows from \Cref{claim:equal-closed-terms-in-A-A'}.
		
		\proofcase{$\xi = t_1 \approx t_2$ with free variables $x_1,\dots, x_r$}
		Follows from \Cref{claim:equal-open-terms-in-A-A'}.
		
		\proofcase{$\xi = \neg \zeta$ without free variables}
		$\zeta$ is also a closed sub-formula of $\varphi$ 
		and from the induction hypothesis:
		
		\begin{align*}
			\A \sat \xi
			&\iff \A \nsat \zeta \\
			&\iff \A' \nsat \mathcal{T}_P(\zeta) \\
			&\iff \A' \sat \mathcal{T}_P(\xi)
		\end{align*}
		
		\proofcase{$\xi = \neg \zeta$ with free variables $x_1,\dots,x_r$}
		$\zeta$ is also a sub-formula of $\varphi$ with free variables 
		$x_1,\dots,x_r$ 
		and from the induction hypothesis, 
		for any $\alpha_1,\dots,\alpha_r \in A$:
		
		\begin{align*}
			\A \sat &\xi \subs \\
			&\iff 
			\A \nsat \zeta \subs \\
			&\iff 
			\A' \nsat \mathcal{T}_P(\zeta) \subs \\
			&\iff 
			\A' \sat \mathcal{T}_P(\xi) \subs
		\end{align*}
		
		\proofcase{$\xi = \zeta_1 \lor \zeta_2$ without free variables}
		$\zeta_1,\zeta_2$ are also closed sub-formulas of $\varphi$, 
		and from the induction hypothesis:
		
		\begin{align*}
			\A \sat \xi
			&\iff \A \sat \zeta_1 \text{ or } \A \sat \zeta_2 \\
			&\iff \A' \sat \mathcal{T}_P(\zeta_1) \text{ or } \A' \sat \mathcal{T}_P(\zeta_2) \\
			&\iff \A' \sat \mathcal{T}_P(\xi)
		\end{align*}
		
		\proofcase{$\xi = \zeta_1 \lor \zeta_2$ with free variables $x_1,\dots,x_r$}
		$\zeta_1,\zeta_2$ are also sub-formulas of $\varphi$ with (at most) free variables 
		$x_1,\dots,x_r$, and from the induction hypothesis, 
		for any $\alpha_1,\dots,\alpha_r$:
		
		\begin{align*}
			\A &\sat \xi \subs \\
			&\iff 
			\A \sat \zeta_1 \subs \\
			&\qquad\text{ or } 
			\A \sat \zeta_2 \subs \\
			&\iff 
			\A' \sat \mathcal{T}_P(\zeta_1) \subs \\
			&\qquad\text{ or } 
			\A' \sat \mathcal{T}_P(\zeta_2) \subs \\
			&\iff 
			\A' \sat \mathcal{T}_P(\xi) \subs
		\end{align*}
		
		\proofcase{$\xi = \forall x.\zeta$ without free variables}
		$\zeta$ is a sub-formula of $\varphi$ with (at most) one free variable $x$. 
		From the induction hypothesis:
		
		\begin{align*}
			\A \sat \xi
			&\iff \text{For any } \alpha \in A.\A \sat \zeta \sub \\
			&\iff \text{For any } \alpha \in A.\A' \sat \mathcal{T}_P(\zeta) \sub \\
			&\iff \A \sat \mathcal{T}_P(\xi)
		\end{align*}
		
		\proofcase{$\xi = \forall x.\zeta$ with free variables $x_1,\dots,x_r$}
		$\zeta$ is a sub-formula of $\varphi$ with (at most) $m+1$ free variables 
		$x,x_1,\dots,x_r$. 
		From the induction hypothesis, 
		for any $\alpha_1,\dots,\alpha_r \in A$:
		
		\begin{align*}
			\A 
			&\sat \xi \subs \\
			&\iff \text{For any }\alpha \in A.\A \sat \zeta  \subs[\dosub{x}{\alpha},]\\
			&\iff \text{For any }\alpha \in A.\A' \sat \mathcal{T}_P(\zeta) \subs[\dosub{x}{\alpha},] \\
			&\iff \A' \sat \mathcal{T}_P(\xi) \subs
		\end{align*}
	\end{proof}
	
	\proofpart{
		If there exists some partition $P$ such that $\mathcal{T}_P(\varphi)$ 
		has a \emph{distinct} \SL{} model, 
		then $\varphi$ has an \SL{} model ($\Leftarrow$)
	}
	Let $\A'$ be some \SL{} model for $\mathcal{T}_P(\varphi)$ 
	and we construct an \SL{} model 
	\[
	\A = \parens{
		A, 
		a^\A_1,\dots, a^\A_l,
		b^\A_m, \dots, b^\A_m,
		c^\A_n, \dots, c^\A_n, 
		s^\A_m, \dots, s^\A_m
	}
	\]
	for $\varphi$, where 
	$A$, 
	$b^\A_m, \dots, b^\A_m$,
	$c^\A_n, \dots, c^\A_n$,
	and $s^\A_m, \dots, s^\A_m$
	are taken from $\A'$.
	
	For every $i \in [1,l]$, $a^\A_i$ is defined as:
	$
	a^\A_i = a'_{P(a_i)}
	$.
	
	\begin{observation}
		$\A$ is a Sum structure, and holds the sum property.
	\end{observation}
	
	\begin{nclaim}\label{claim:equal-closed-terms-in-A-A'-2}
		For any closed term $t$, 
		$\interp(t) = \interp'(\mathcal{T}_P(t))$
	\end{nclaim}
	\begin{proof}
		Similarly to \Cref{claim:equal-closed-terms-in-A-A'}, 
		we only need to consider $t = a_i$, and we get:
		
		\begin{align*}
			\interp(t)
			&= \interp(a_i) \\
			&= a^\A_i \\
			&= a'_{P(a_i)} \\
			&= \interp' ( a_{ P(a_i) } ) = \interp' ( \mathcal{T}_P(t) )
		\end{align*}
	\end{proof}\begin{nclaim}\label{claim:equal-open-terms-in-A-A'-2}
		For any term $t$ with free variables $x_1,\dots, x_r$, 
		for any assignment $\Delta$, 
		and for any $\alpha_1,\dots,\alpha_r \in A$, 
		we define $\Delta' = \Delta \subs$, and
		\[
		\interp_{\Delta'}(t \subs)
		=
		\interp'_{\Delta'}(\mathcal{T}_P(t) \subs)\;. \]
	\end{nclaim}
	\begin{proof}
		Identical to the proof of \Cref{claim:equal-closed-terms-in-A-A'-2}.
	\end{proof}
	
	\begin{nclaim}\label{claim:distinct-sat-2}
		Let $\xi$ be a sub-formula of $\varphi$, therefore:
		
		\begin{enumerate}
			\item If $\xi$ is a closed formula, 
			then $\A' \sat \mathcal{T}_P(\xi) \iff \A \sat \xi$.
			
			\item If $\xi$ is a formula with free variables 
			$
			x_1, \dots, x_r
			$ then for every 
			$
			\alpha_1, \dots, \alpha_r \in A
			$,
			
			$
			\A' \sat \mathcal{T}_P(\xi) 
			\subs
			\iff
			\A \sat \xi 
			\subs
			$
		\end{enumerate}
	\end{nclaim}
	
	\noindent
	Since $\varphi$ is a closed formula, 
	and since $\A \sat \mathcal{T}_P(\varphi)$ 
	we get that $\A' \sat \varphi$.
	
	\begin{proof}[\Cref{claim:distinct-sat-2}]
		In the same vain of \Cref{claim:distinct-sat-1}, 
		this follows from \Cref{claim:equal-closed-terms-in-A-A'-2,claim:equal-open-terms-in-A-A'-2}.
	\end{proof}
\end{longproof}

\subsection*{Small \Addr{} Space Proof}
\smallModelProperty*

\smallModelsForSingle*
\begin{longproof}{\Cref{thm:small-models-for-single}}
Let there be some universal, 
closed formula $\varphi$ over 
$\Sigma = \Sigma^{l,1,d}_{\cancel{+},\cancel{\leq}}$ 
and let there be some minimal structure 
$\A \in \structure{\Sigma}$ such that
$\A \sat_\text{SL} \varphi$ 
(i.e. $\A$ is an \SL{} model for $\varphi$).

We denote the (finite) size of $\A$ as $z \triangleq \abs{A}$, 
and we assume towards contradiction that 
$z \geq l+\abs{\varphi}+1$ 
(as our bound function is $\kappa(x) = l+x+1$). 
We construct a smaller model $\A'$ for $\varphi$. 
Thus contradicting the minimality of $\A$, 
and proving our desired claim.

We write out the given model
$
\A = \parens{
	A,
	a^\A_1, \dots, a^\A_l,
	b^\A,
	c^\A_1, \dots, c^\A_n,
	s^\A	
}
$

We know that $\abs{A} = z > l$, and therefore the set 
\[
S \triangleq A \setminus \braces{ 
	a^\A_1, \dots,  a^\A_l
}
\]
is not empty. Let us define 
\[
\alpha^* \triangleq \argmin_{\alpha \in S}\braces{b^\A(\alpha)}
\]
and $b^* \triangleq b^\A(\alpha^*)$.
We construct the  \emph{smaller} \SL{} structure \linebreak
$
\A' = \parens{
	A', a^{\A'}_1, \dots, a^{\A'}_l, b^{\A'}, c^{\A'}_1, \dots, c^{\A'}_n, s^{\A'}	
}
$, where
\begin{align}
	A' 
	&\triangleq 
	A \setminus \braces{ \alpha^* } \\
	a^{\A'}_i 
	&\triangleq 
	a^\A_i \\
	b^{\A'} 
	&\triangleq 
	b^\A \text{ projected on }A' \\
	s^{\A'} &\triangleq s^\A - b^*
\end{align}

and we postpone defining $c^{\A'}_k$ for now. We observe that:

\begin{observation}
	If $\A$ is a \emph{distinct} \SL{} model, then so is $\A'$.
\end{observation}

~\\
\noindent
Firstly, we prove the following claim:

\begin{nclaim}
	$\A'$ holds the sum property.
\end{nclaim}
\begin{proof}
	Since $\A$ holds the sum property for $s^\A$:
	\begin{align*}
		s^{\A'} 
		&= s^\A - b^* \\
		&= \sum_{\alpha \in A} b^\A(\alpha) - b^* \\
		&= 
		\parens{
			\sum_{\alpha \in A \setminus \braces{\alpha^*}} b^\A(\alpha)
		} 
		+ 
		\underbrace{b^\A(\alpha^*) }_{ = b^*}
		-
		b^* \\
		&= \sum_{\alpha \in A'} b^{\A'}(\alpha)
	\end{align*}
\end{proof}

The definition for $c^{\A'}_k$ depends on $b^*$. 
If $b^* = 0$ then simply 
$c^{\A'}_k = c^\A_k$. 
In this case, we prove the following:

\begin{lemma}
	\label{lem:small-with-zero}
	For any term $t$, assignment $\Delta$,
	\[
	\interp_\Delta(t) = \interp'_\Delta(t)
	\]
\end{lemma}

\begin{proof}
	Since $b^* = 0$, we get that $s^{\A'} = s^\A$ 
	and therefore the interpretations of $\A$ and $\A'$ are identical 
	--- $\interp = \interp'$. 
\end{proof}

\begin{corollary}
	Since the domain of $\A'$ is a strict subset of the domain of $\A$, 
	for any formula $\xi$, 
	$\A \sat \xi \To \A' \sat \xi$, 
	and in particular $\A'$ is also an \SL{} model for $\varphi$.
\end{corollary}

\noindent
In the case that $b^* > 0$, we define
\[
c^{\A'}_k 
\triangleq 
\left\{\begin{array}{ll}
	s^\A - b^* & \mbox{if } c^\A_k = s^\A \\
	c^\A_k + 1 & \mbox{if } c^\A_k \geq s^\A - b^* \text{ and } c^\A_k \neq s^\A \\
	c^\A_k & \mbox{otherwise}
\end{array}\right.
\]
and the proof is more involved. We firstly make the following observations:

\begin{observation}
	\label{obs:small-const-vs-sum}
	For any $k \in [1,d]$, 
	\[
	c^\A_k = s^\A \iff c^{\A'}_k = s^{\A'}
	\]
\end{observation}

\begin{observation}
	\label{obs:small-const-vs-const}
	For any $k_1,k_2 \in [1,d]$, 
	\[
	c^\A_{k_1} = c^\A_{k_2} \iff c^{\A'}_{k_1} = c^{\A'}_{k_2}
	\]
\end{observation}

\noindent
The central claim we need to prove is:
\begin{nclaim}
	\label{claim:small-sat}
	Let $\xi$ be a sub-formula of $\varphi$,
	\begin{enumerate}
		\item If $\xi$ is a closed, quantifier-free formula then 
		\[
		\A \sat \xi \iff \A' \sat \xi
		\]
		
		\item If $\xi$ is a quantifier-free formula with free variables $x_1, \dots, x_r$, 
		then for every \linebreak $\alpha_1, \dots, \alpha_r \in A'$,
		\[
		\A \sat \xi \subs
		\iff
		\A' \sat \xi \subs
		\]
		
		\item If $\xi$ is a closed, universally quantified formula then 
		\[
		A \sat \xi \To \A' \sat \xi
		\]
		
		\item If $\xi$ is a universally quantified formula with free variables $x_1, \dots, x_r$, 
		then for every $\alpha_1, \dots, \alpha_r \in A'$,
		\[
		\A \sat \xi \subs
		\To
		\A' \sat \xi \subs
		\]
	\end{enumerate}
\end{nclaim}

\begin{corollary}
	$\A' \sat \varphi$.
\end{corollary}
\begin{proof}
	Since $\varphi$ is a closed, universally quantified sub-formula of itself, ]
	and since it is given that $\A \sat \varphi$, 
	we get from \Cref{claim:small-sat} that $\A' \sat \varphi$.
\end{proof}

In order to prove \Cref{claim:small-sat} we firstly need to prove the following two lemmas:

\begin{lemma}
	\label{lem:small-balances-vs-sum}
	For any $\alpha \in A'$,
	\[
	b^\A(\alpha) = b^{\A'}(\alpha) < s^{\A'} < s^\A
	\]
\end{lemma}
\begin{proof}
	First, since $b^{\A'}$ is defined to be a projection of $b^\A$ on 
	a subset of its domain $A' \subseteq A$ 
	it is obvious that $b^\A(\alpha) = b^{\A'}(\alpha)$ 
	for any $\alpha \in A'$.
	
	Also, since $s^{\A'} = s^\A - b^*$ 
	and we know that $b^* > 0$, 
	it is clear that $s^{\A'} < s^\A$.
	
	What remains to prove is that for any $\alpha \in A'$, 
	$b^{\A'}(\alpha) < s^{\A'}$. 
	$A'$ has at least $l+\abs{\varphi}+1$ elements, 
	and therefore 
	$A' \setminus \braces{ a^{\A'}_1, \dots, a^{\A'}_l }$ 
	has at least 2 elements. 
	Let us denote them: $\alpha_1, \alpha_2$.
	
	For both of these elements, 
	\[
	b^{\A'}(\alpha_1), b^{\A'}(\alpha_2) > 0
	\]
	since otherwise they would have been chosen as $\alpha^*$ 
	--- contradicting $b^*$'s minimality.
	
	\noindent
	For any element $\alpha$, since $\A'$ holds the sum property,
	\begin{align*}
		s^{\A'} &= \sum_{ \alpha' \in A' } b^{\A'}(\alpha')
		\\
		&= b^{\A'}(\alpha) + \sum_{ \alpha' \in A' \setminus \braces{\alpha} } b^{\A'}(\alpha')
	\end{align*}
	
	We can re-arrange and get that
	\[
	b^{\A'}(\alpha) = s^{\A'} - \sum_{\alpha' \in A' \setminus \braces{\alpha} } b^{\A'}(\alpha')
	\]
	and since $A' \setminus \braces{ \alpha }$ contains either $\alpha_1$ or $\alpha_2$, it must be that
	\[
	\sum_{\alpha' \in A' \setminus \braces{\alpha} } b^{\A'}(\alpha') > 0
	\]
	and therefore $b^{\A'}(\alpha) < s^{\A'}$.
\end{proof}

\begin{lemma}
	\label{lem:small-sum-vs-formula-length}
	\[
	\abs{\varphi} < s^{\A'} < s^\A
	\]
\end{lemma}
\begin{proof}
	Let us examine the set $S \triangleq A' \setminus \braces{a^{\A'}_1, \dots, a^{\A'}_l}$. 
	It has at least $\abs{\varphi}+1$ elements.
	
	For any $\alpha \in S$, $b^{\A'}(\alpha) > 0$, 
	otherwise it would have been chosen as $\alpha^*$ and we'd have $b^* = 0$ 
	--- which contradicts $b^*$'s minimality.
	
	Therefore, on the one hand,
	\[
	\sum_{\alpha \in S} b^{\A'}(\alpha) \geq \abs{S} \geq \abs{\varphi}+1 > \abs{\varphi}
	\]
	
	And, on the other hand, since $S \subseteq A'$, we know that
	\[
	\sum_{\alpha \in S} b^{\A'}(\alpha) \leq \sum_{\alpha \in A'} b^{\A'}(\alpha) = s^{\A'}
	\]
	
	And combining the two results we get that 
	$\abs{\varphi} < s^{\A'}$, 
	and since $b^* > 0$, $s^{\A'} < s^\A$.
\end{proof}

\begin{proof}[Proof of \Cref{claim:small-sat}]
We prove the claim using structural induction.

\proofstep{$\xi = t_1 \approx t_2$ without free variables}
$\xi$ is a closed, quantifier-free formula, 
so we prove that 
$\A \sat \xi \iff \A' \sat \xi$. 
We consider the following cases:

\proofcase{$t_1 = t_2$} Tautology.

\proofcase{$t_1 = s, t_2 = \text{numeral}$}	
Since $\xi$ is a sub-formula of 
$\varphi$, $\abs{\xi} \leq \abs{\varphi}$, 
and therefore the numeral is less than $\abs{\varphi}$.

However, 
$s^\A, s^{\A'} > \abs{\varphi}$ 
from 
\Cref{lem:small-sum-vs-formula-length} 
and therefore $\A, \A' \nsat \xi$.

\proofcase{$t_1 = s, t_2 = c_k$}
From Observation \ref{obs:small-const-vs-sum} we know that 
$
s^\A = c^\A_k \iff s^{\A'} = c^{\A'}_k
$ 
and therefore $\A \sat \xi \iff \A' \sat \xi$.

\proofcase{$t_1 = s, t_2 = b(a_i)$}
From \Cref{lem:small-balances-vs-sum} 
we know that for any $\alpha \in A'$, 
$b^{\A'}(\alpha) = b^\A(\alpha) < s^{\A'} < s^\A$ 
and in particular for 
$\alpha = a^{\A}_i = a^{\A'}_i$, $\A, \A' \nsat \xi$.

\proofcase{$t_1 = c_k, t_2 = \text{numeral}$}
If $c^\A_k = c^{\A'}_k$ then trivially 
$\A \sat \xi \iff \A' \sat \xi$.

Otherwise, 
$c^\A_k , c^{\A'}_k \geq s^{\A'}$. 
However, since $\xi$ is a sub-formula of $\varphi$, 
the numeral is less than $\abs{\varphi}$, 
and $s^{\A'} > \abs{\varphi}$, 
from \Cref{lem:small-sum-vs-formula-length}. 
Therefore, $\A, \A' \nsat \xi$.

\proofcase{$t_1 = c_{k_1}, t_2 = c_{k_2}$}
Trivial, from Observation \ref{obs:small-const-vs-const}.

\proofcase{$t_1 = c_k, t_2 = b(a_i)$}
If $c^\A_k = c^{\A'}_k$ then 
from \Cref{lem:small-balances-vs-sum}, 
$\A \sat \xi \iff \A' \sat \xi$.

Otherwise, $c^\A_k , c^{\A'}_k \geq s^{\A'}$. 
However, 
from \Cref{lem:small-balances-vs-sum} 
we know that for any $a \in A'$ 
(and in particular for $a_j^\A$), 
$b^{\A'}(a) = b^\A(a) < s^{\A'} \leq c^\A_k, c^{\A'}_k$. 
Therefore, $\A, \A' \nsat \xi$.

\proofcase{$t_1 = a_{i_1}, t_2 = a_{i_2}$}
Trivial, since the interpretation of the 
\Addr{} constants is identical in $\A$, $\A'$.

Any other case is symmetrical to one of the cases above.

\proofstep{$\xi = t_1 \approx t_2$ with free variables $x_1,\dots,x_r$}
$\xi$ is a quantifier-free formula, 
so we prove that for any $\alpha_1, \dots, \alpha_r \in A'$
\[
\A \sat \xi \subs
\iff
\A' \sat \xi \subs
\]

We consider the following cases:
\proofcase{$t_1 = t_2$}
Tautology.

\proofcase{$t_1 = s, t_2 = b(x)$}
From \Cref{lem:small-balances-vs-sum} we know that for any 
$\alpha \in A'$, 
$b^{\A'}(\alpha) = b^\A(\alpha) < s^{\A'} < s^\A$ 
and in particular 
$\A, \A' \nsat \xi \sub$.

\proofcase{$t_1 = b(x), t_2 = \text{numeral}$}
Trivial, since 
from \Cref{lem:small-balances-vs-sum}, 
for every $\alpha \in A'$, 
$b^\A(\alpha) = b^{\A'}(\alpha)$.

\proofcase{$t_1 = b(x), t_2 = c_k$}
Let there be $\alpha \in A'$, 
we separate into the following cases:
\begin{enumerate}
	\item If $c^\A_k = c^{\A'}_k$:
	
	From \Cref{lem:small-balances-vs-sum} we get
	\begin{align*}
		\A \sat \xi \sub
		&\iff \A \sat b(\alpha) \approx c_k \\
		&\iff b^\A(\alpha) = c^\A_k \\
		&\iff b^{\A'}(\alpha) = c^{\A'}_k \\
		&\iff \A' \sat b(\alpha) \approx c_k \\
		&\iff \A' \sat \xi \sub
	\end{align*}
	
	\item Otherwise, $c^\A_k \geq s^{\A'}$ 
	and $c^{\A'}_k \geq s^{\A'}$. 
	From \Cref{lem:small-balances-vs-sum} 
	we get $b^{\A'}(a) = b^\A(a) < s^{\A'} < s^\A$ 
	and therefore 
	$\A \nsat \xi \sub$ 
	and 
	$\A' \nsat \xi \sub$.
\end{enumerate}

\proofcase{$t_1 = b(x), t_2 = b(a_i)$}
Trivial from \Cref{lem:small-balances-vs-sum}.

\proofcase{$t_1 = b(x_1), t_2 = b(x_2)$}
Trivial from \Cref{lem:small-balances-vs-sum}.

\proofcase{$t_1 = a_i, t_2 = x$}
Trivial, since the interpretation of the address constants is identical 
in $\A$ and $\A'$.

\proofcase{$t_1 = x_1, t_2 = x_2$}
Trivially holds for any $a \in A'$.

Any other case is symmetrical to one of the cases above.

\proofstep{$\xi = \neg \zeta$ without free variables}
Since $\varphi$ is a universal formula we can assume it is in prenex form, 
and therefore, $\zeta$ is a closed, quantifier-free formula, shorter than $\xi$ 
and from the induction hypothesis,
$
\A \sat \zeta \iff \A' \sat \zeta
$,
and therefore
\begin{align*}
	\A \sat \xi
	&\iff \A \nsat \zeta \\
	&\iff \A' \nsat \zeta \\
	&\iff \A' \sat \xi
\end{align*}

\proofstep{$\xi = \neg \zeta$ with free variables $x_1,\dots,x_r$}
Similarly to the closed formula case, 
$\zeta$ is a quantifier-free formula with free variables 
$x_1,\dots,x_r$ 
and the claim holds from the induction hypothesis.

\proofstep{$\xi = \zeta_1 \lor \zeta_2$ without free variables}
Similarly to the negation formula case, 
$\zeta_1,\zeta_2$ are closed, quantifier-free formulas 
and the claim holds from the induction hypothesis.

\proofstep{$\xi = \zeta_1 \lor \zeta_2$ with free variables $x_1, \dots, x_r$}
Similarly to the closed formula case, 
and the claim holds from the induction hypothesis.

\proofstep{$\xi = \forall v.\zeta$ without free variables}
Since $\xi$ is a universal formula, 
we need to show that if $\A \sat \xi$ 
then $\A' \sat \xi$ (but not vice versa).

$\zeta$ is a universally quantified formula with (at most) one free variable $x$. 
If $\A \sat \xi$ then for every $\alpha \in A$,
\[
\A \sat \zeta \sub
\]

and in particular for any $\alpha \in A' \subsetneq A$.

$\zeta$ is shorter than $\xi$ and therefore the induction hypothesis holds:
\[
\A' \sat \zeta \sub
\]

for any $\alpha \in A'$, and therefore $\A' \sat \xi$.

\proofstep{$\xi = \forall v.\zeta$ with free variables $x_1, \dots, x_r$}
Let there be $\alpha_1, \dots, \alpha_r \in A'$. If
\[
\A \sat \xi 
\subs
\]

then for every $\alpha \in A$,
\[
\A \sat \zeta 
\subs[\dosub{x}{\alpha},]
\]

and in particular for every $\alpha \in A' \subsetneq A$. 
From the induction hypothesis for $\zeta$ we get:
\[
\A' \sat \zeta 
\subs[\dosub{x}{\alpha},]
\]

which is true for any $a \in A'$, and therefore,
\[
\A' \sat \xi 
\subs
\]
\end{proof}
\end{longproof}

\subsection*{Presburger Reduction Proof}
\subsubsection{Defining the Transformations}
The transformation of formulas from SL to PA works by 
explicitly writing out sums as additions and universal quantifiers as conjunctions. 
Since we're dealing with a fragment of SL that has some bound function $\kappa(\cdot)$, 
we know that for given formula $\varphi$, 
there is a model with at most $\kappa(\abs{\varphi})$ elements of \Addr{} sort.

Moreover, we use 
$
\tilde \kappa \define 
\max \braces{ 
	\kappa(\abs{\varphi}), 
	l 
}
$ 
as the upper bound 
(where $l$ is the amount of \Addr{} constants). 
Since we're looking for distinct models, 
it is obvious that we need at least $l$ distinct elements.

For each balance function $b^1_j$ we have $\tilde \kappa$ constants 
$b_{1,j}, \dots, b_{\tilde \kappa,j}$.

In addition we have $\tilde \kappa$ indicator constants 
$a_1, \dots, a_{\tilde \kappa}$, 
to mark if an \Addr{} element is "active". 
An inactive element has all zero balances, 
and is skipped over in universal quantifiers.

Any \Addr{} constant $a_i$ or \Addr{} variable $x$ is handled in two ways, 
depending on the context they appear in:
\begin{itemize}
	\item If they are compared, we replace the comparison with $\top$ or $\bot$; 
	we know statically if the comparison holds, 
	since the \Addr{} constants are distinct and 
	every universal quantifier is written out as a conjunction.
	
	\item Otherwise, they must be used in some balance function $b^1_j$, 
	and then they are substituted with the corresponding $b_{i,j}$ or $b_{x,j}$ 
	(which will be determined once the universal quantifiers are unrolled).
\end{itemize}

The integral constants $c_1, \dots, c_d$ are simply copied over.

In summary:

\begin{definition}[Corresponding Presburger Vocabulary]
	\label{def:pres-vocab}
	Given the \SL{} vocabulary $\Sigma^{l,m,d}$ and a bound $\tilde \kappa \geq l$, 
	we define the \emph{corresponding Presburger vocabulary} as
	\[
	\Sigma^{l,m,d}_{\text{Pres}(\tilde \kappa)} = \text{Pres}(\Sigma^{l,m,d}, \tilde \kappa) 
	\define
	\parens{
		a_1, \dots, a_{\tilde \kappa},
		b_{1,1}, \dots, b_{\tilde \kappa,m},
		c_1, \dots, c_d,
		0, 1,
		+^2, \leq^2
	}\;. \]
\end{definition}

Firstly, we define the simpler auxiliary formula $\eta(\varphi)$ in three parts:
\begin{definition}
	\label{def:aux-1}
	We require that inactive \Addr{} elements have zero balances ---
	\[
	\eta_1(\varphi) = 
	\bigland_{i=1}^{\tilde \kappa}
	\brackets{	
		\parens{ a_i \approx 0 }
		\to 
		\parens{ \bigland_{j=1}^m b_{i,j} \approx 0 }
	}
	\]
\end{definition}

\begin{definition}
	\label{def:aux-2}
	And that elements referred by \Addr{} constants be active ---
	\[
	\eta_2(\varphi) =
	\bigland_{i=1}^l a_i \not \approx 0
	\]
\end{definition}

\begin{definition}
	\label{def:aux-3}
	Finally, we require that the active elements are a continuous sequence starting at 1. 
	Or, put differently, once an indicator is zero, all indicators following it are also zero:
	\[
	\eta_3(\varphi) = \bigland_{i=1}^{\tilde \kappa} 
	\brackets{
		a_i \approx 0 
		\to
		\parens{\bigland_{i'=i}^{\tilde \kappa} a_{i'} \approx 0}
	}
	\]
\end{definition}

\noindent
The complete auxiliary formula is then 
$
\eta(\varphi) = \eta_1(\varphi) \land \eta_2(\varphi) \land \eta_3(\varphi)
$.

In order to define $\tau(\varphi)$, 
we firstly define the transformation for terms, 
and then build up the complete transformation, using several substitutions:

\begin{definition}
	\label{def:trans-0}
	We define the terms transformation inductively, 
	and we substitute balances and \Addr{} terms (constants or variables) 
	with placeholders (marked with *), which are further substituted:
	\[
	\tau_0(t) 
	\define \left\{
	\begin{array}{ll}
		a_i^* 
		& \mbox{if } t = a_i \\
		x^* 
		& \mbox{if } t = x \text{ for some free variable} \\
		b_{1,j} + \dots+ b_{\tilde \kappa,j}
		& \mbox{if } t = s_j \\
		b^*_j(\tau_0(t_1)) 
		& \mbox{if } t = b_j(t_1) 
		\text{ where }
		t_1 \in \braces{ a_i, x }
		\\
		\tau_0(t_1) + \tau_0(t_2) 
		& \mbox{if } t = t_1 + t_2
	\end{array}
	\right.
	\]
\end{definition}

\begin{definition}
	\label{def:trans-1}
	Next we define the transformation for formulas, 
	replacing only variable placeholders:
	\[
	\tau_1(\xi)
	\define \left\{
	\begin{array}{ll}
		\tau_0(t_1) \approx \tau_0(t_2) 
		& \mbox{if } \xi = t_1 \approx t_2 \\
		\tau_0(t_1) \leq \tau_0(t_2)
		& \mbox{if } \xi = t_1 \leq t_2 \\
		\neg \tau_1(\zeta)
		& \mbox{if } \xi = \neg \zeta \\
		\tau_1(\zeta_1) \land \tau_1(\zeta_2)
		& \mbox{if } \xi = \zeta_1 \land \zeta_2 \\
		\bigland_{i = 1}^{\tilde \kappa} 
		\parens{
			a_i \approx 0
			\lor
			\tau_1(\zeta) \sub[x^*][a^*_i]
		}
		& \mbox{if } \xi = \forall x. \zeta
	\end{array}
	\right.
	\]
\end{definition}

\noindent
We can see that for any formula $\xi$ containing arbitrary terms, 
$\tau_1(\xi)$ only has $a^*_i$ and $b^*_j$ placeholders (but no $x^*$ ones).

\begin{definition}
	\label{def:sub-address-comparisons}
	Now we define a substitution $\sigma_1$ that removes 
	\Addr{} comparisons by evaluating them:
	\[
	\sigma_1 \define 
	\sub[\parens{ a^*_i \approx a^*_i }][\top]
	\sub[ \parens{ a^*_i \approx a^*_{i'} }][\bot]
	\]
	where $i,i' \in [1,\tilde \kappa]$.
\end{definition}

\begin{note}
	We first replace comparisons where $a^*_i \approx a^*_i$, 
	which is equivalent to \textbf{true} ($\top$). 
	Then any remaining comparison must be where $i \neq i'$, 
	and therefore equivalent to \textbf{false} ($\bot$).
\end{note}

\begin{definition}
	\label{def:sub-balance-placeholders}
	Finally, we're left with placeholders inside balance functions, 
	which we substitute by the corresponding balance constant:
	\[
	\sigma_2 \define
	\sub[b^*_j(a^*_i)][b_{i,j}]
	\]
	where $i \in [1,\tilde \kappa], j \in [1,m]$.
\end{definition}

\begin{definition}
	\label{def:trans-3}
	The complete transformation is then:
	\[
	\tau(\varphi) \define \tau_1(\varphi) \sigma_1 \sigma_2
	\]
\end{definition}

\noindent
Given the above definitions, 
let us recall the Presburger Reduction Theorem:

\presburgerReductionTheorem*

\begin{longproof}{\Cref{thm:presburger-reduction}}
We first define congruence between \SL{} structures 
and structures over the corresponding Presburger vocabulary, 
and we prove a general theorem about them. 
We use that congruence theorem to prove that a formula $\varphi$ has 
a \emph{distinct} \SL{} model with \emph{small \Addr{} space}
iff 
$\varphi'$ has a Standard Model of Arithmetic.
~\\

\proofpart{Congruence Lemmas}

\begin{definition}
	\label{def:SL-PA-congruence}
	
	Given an \SL{} vocabulary $\Sigma^{l,m,d}$, 
	a bound $\tilde \kappa \geq l$ and a formula $\varphi$ over $\Sigma^{l,m,d}$. 
	Then 
	$\A \in \structure{\Sigma^{l,m,d}}$ and 
	$\A' \in \structure{\text{Pres}(\Sigma^{l,m,d}, \tilde \kappa)}$ 
	are said to be congruent if the following conditions hold:
	
	\begin{enumerate}
		\item 
			\label[cong]{item:congruence-sum}
			$\A$ holds the sum property.
			
		\item 
			\label[cong]{item:congruence-sat-eta}
			$\A'$ satisfies $\eta(\varphi)$.
			
		\item 
			\label[cong]{item:congruence-bounded-size}
			$z \triangleq \abs{A} \leq \tilde \kappa$, and we write out $A = \braces{ \alpha_1, \dots, \alpha_z }$.
			
		\item 
			\label[cong]{item:congruence-address-constants}
			For any $i \in [1,l]$, $a^\A_i = \alpha_i$.
			
		\item 
			\label[cong]{item:congruence-balances}
			For any $j \in [1,m]$, for any $i \in [1,z]$, $b^{\A'}_{i,j} = b^\A_j(\alpha_i)$, and for any $i > z$, $b^{\A'}_{i,j} = 0$.
			
		\item 
			\label[cong]{item:congruence-indicators}
			For any $i \in [1,z]$, $a^{\A'}_i > 0$ and for any $i > z$, $a^{\A'}_i = 0$.
			
		\item 
			\label[cong]{item:congruence-distinct}
			$\A$ is distinct, and in particular $l \leq z$.
	\end{enumerate}
\end{definition}

\begin{lemma}
	\label{lem:congruent-nat-closed-terms}
	Let $\A,\A'$ be two congruent structures for \SL{} vocabulary $\Sigma$, bound $\tilde \kappa$ and formula $\varphi$.
	For any ground term $t$ of sort \Nat{} over $\Sigma$,
	\[
	\interp'(\tau_0(t) \sigma_2) = \interp(t)
	\]
\end{lemma}
\begin{proof}
	We prove the lemma using structural induction over all possible ground terms:
	
	\proofstep{$t = s_j$ for any $j \in [1,m]$}
	From \Cref{item:congruence-sum} $\A$ holds the sum property, and therefore:
	\[
	\interp(s_j) 
	= s^\A_j
	= \sum_{\alpha \in A} b^\A_j(\alpha)
	= \sum_{i = 1}^z b^\A_j(\alpha_i)
	\]
	From \Cref{item:congruence-balances}, for any $i \in [1,z]$, 
	$
	b^{\A'}_{i,j} = b^\A_j(\alpha_i)
	$, and for any $i \in [z + 1, \tilde \kappa]$,
	$
	b^{\A'}_{i,j} = 0
	$, therefore we can write the sum above as
	\[
	\interp(s_j) 
	= \dots
	= \sum_{i = 1}^{\tilde \kappa} b^{\A'}_{i,j}
	\]
	
	From the definition of $\tau_0$ we get:
	\[
	\interp'(\tau_0(s_j) \sigma_2)
	= 
	\interp' \parens{
		\brackets{b_{1,j} + \dots + b_{\tilde \kappa, j}} \sigma_2
	}
	= \sum_{i = 1}^{\tilde \kappa} b^{\A'}_{i,j}
	\]
	Since we have no placeholders, $\sigma_2$ has no effect, and we get the same expression as for $\interp(t)$.
	
	\proofstep{$t = b_j(a_i)$ where $i \in [1,l], j \in [1,m]$}
	From \Cref{item:congruence-address-constants}, $a^\A_i = \alpha_i$, and we get:
	\[
	\interp(t) = b^\A_j(\alpha_i)
	\]
	
	From the definition of $\tau_0$ and $\sigma_2$ we get:
	\[
	\tau_0(t) \sigma_2
	= \brackets{b^*_j(a^*_i)} \sigma_2
	= b_{i,j}
	\]
	
	And therefore, since $\A$ is distinct, $i \leq l \leq z$, and from \Cref{item:congruence-balances},
	\[
	\interp'(\tau_0(t) \sigma_2) = \interp'(b_{i,j}) = b^{\A'}_{i,j} = b^\A_j(\alpha_i)
	\]
	
	\proofstep{$t = t_1 + t_2$}
	Follows from the induction hypothesis for $t_1$ and $t_2$ (since $+$ is interpreted in the same way in $\A$ and $\A'$).
\end{proof}

\begin{lemma}
	\label{lem:congruent-nat-open-terms}
	Let $\A,\A'$ be two congruent structures for \SL{} vocabulary $\Sigma$, bound $\tilde \kappa$ and formula $\varphi$.
	For any term $t$ with at most $r$ free variables $x_1, \dots, x_r$, for any indices $i_1, \dots, i_r \in [1,z]$, for any assignment $\Delta$ we define
	\[
	\Delta' = \Delta 
	\subsi
	\]
	
	and the following holds:
	\[
	\interp_{\Delta'}(t)
	=
	\interp' \parens{	
		\tau_0(t) 
		\subsi[][r][x^*][a^*]
		\sigma_2
	}
	\]
\end{lemma}
\begin{proof}
	We prove the lemma using structural induction over all possible terms with free variables:
	
	\proofstep{$t = b_j(x)$}
	For any $i \in [1,z]$,
	\[
	\interp_{\Delta'}(t) = \interp(b_j) \parens{ \Delta'(x) } = b^\A_j(\alpha_i)
	\]
	
	By definition, $\tau_0(t) = b^*_j(x^*)$, and therefore
	\begin{align*}
	\interp' \parens{ \tau_0(t) \sub[x^*][a^*_i] \sigma_2 }
		&= \interp' \parens{ b^*_j(a^*_i) \sigma_2 }
	\\
		&= \interp' \parens{ b_{i,j} }
	\\
		&= b^{\A'}_{i,j}
	\\
		&= b^\A_j(\alpha_i)
	\end{align*}
	since $i \in [1,z]$.
	
	\proofstep{$t = t_1 + t_2$}
	Either $t_1$ or $t_2$ has free variables, and let us assume w.l.o.g. that $t_1$ does. Therefore, $t_2$ is either a ground term, or also has free variables. If $t_2$ has no free variables, the substitution of free variables wouldn't affect it.
	
	In both cases, from the induction hypothesis and from \Cref{lem:congruent-nat-closed-terms}, for any $i_1,\dots,i_r \in [1,z]$,
	\[
	\interp_{\Delta'}(t_v)
	=
	\interp'(\tau_0(t_v)
	\subsi[][r][x^*][a^*]
	)
	\]
	where $v \in \braces{1,2}$, and we get desired equality for $t$ as well.
\end{proof}

\begin{lemma}
	\label{lem:congruent-sat}
	Let $\A,\A'$ be two congruent structures for \SL{} vocabulary $\Sigma$, bound $\tilde \kappa$ and formula $\varphi$.
	Let $\xi$ be a sub-formula of $\varphi$, therefore:
	\begin{enumerate}
		\item If $\xi$ is a closed formula, then $\A' \sat \tau(\xi) \iff \A \sat \xi$.
		
		\item If $\xi$ is a formula with free variables $x_1, \dots, x_r$ then for any $i_1, \dots, i_r \in [1,z]$:
		\begin{align*}
		\A' 
		&\sat \tau_1(\xi) 
		\subsi[][r][x^*][a^*]
		\sigma_1 \sigma_2  
		\\
		&\iff \A \sat \xi 
		\subsi[][r][x][\alpha]
		\end{align*}
	\end{enumerate}
\end{lemma}

\begin{proof}
	Let us separate into the following steps:
	\proofstep{$\xi = t_1 \approx t_2$ without free variables, where $t_1,t_2$ are of Address sort}
	
	Since $t_1,t_2$ are Addresses, there are two indices $i_1,i_2 \in [1,l]$ such that $t_1 = a_{i_1}, t_2 = a_{i_2}$. From \Cref{item:congruence-distinct}, $\A$ is distinct and therefore
	\[
	\A \sat \xi \iff i_1 = i_2.
	\]
	
	As for $\tau(\xi)$,
	\begin{align*}
	\tau(\xi) 
		&=
		\parens{ 
			\tau_0(t_1) \approx \tau_0(t_2) 
		} \sigma_1 \sigma_2
	\\
		&= 
		\parens{ 
			a^*_{i_1} \approx a^*_{i_2} 
		} \sigma_1
	\\
		&= \ifelsedef{\top}{i_1 = i_2}{\bot}
	\end{align*}

	Which means that $\A' \sat \tau(\xi) \iff i_1 = i_2 \iff \A \sat \xi$.
	
	\proofstep{$\xi = t_1 \approx t_2$ without free variables, where $t_1,t_2$ are of sort \Nat{}}
	In this case, $\sigma_1$ would not change the formula and we can apply $\sigma_2$ to each term:
	\[
	\tau(\xi) = \tau_0(t_1) \sigma_2 \approx \tau_0(t_2) \sigma_2
	\]
	
	Since $t_1,t_2$ are of sort \Nat{}, from \Cref{lem:congruent-nat-closed-terms} we get that
	\begin{align*}
	\A' \sat \tau(\xi) 
	&\iff 
	\interp'(\tau_0(t_1) \sigma_2) 
	= 
	\interp'(\tau_0(t_2) \sigma_2)
	\\
	\tag{\Cref{lem:congruent-nat-closed-terms}}
	&\iff
	\interp(t_1)
	=
	\interp(t_2)
	\\
	&\iff
	\A \sat t_1 \approx t_2 = \xi
	\end{align*}
	
	\proofstep{$\xi = t_1 \approx t_2$ with free variables $x_1, \dots, x_r$, where $t_1,t_2$ are of Address sort}	
	Let us first define $\sigma = \subsi$, $\sigma' = \subsi[][r][x^*][a^*]$.
	
	If $t_1,t_2$ both have free variables then we can write them as $t_1 = x_1, t_2 = x_2$ and after substituting $\sigma$ we get that
	
	\[
	\xi \sigma = \alpha_{i_1} \approx \alpha_{i_2}.
	\]
	
	Therefore, $\A \sat \xi \sigma \iff i_1 = i_2$.
	
	As for $\A'$, we get
	\begin{align*}
	\tau_1(\xi) \sigma \sigma_1 \sigma_2 
	&= 
		\parens{ 
			x^*_1 \approx x^*_2 
		} \sigma' \sigma_1 \sigma_2
	\\
	&=
		\parens{ 
			a^*_{i_1} \approx a^*_{i_2}
		} \sigma_1 \sigma_2
	\\
	&=
	\parens{ 
		a^*_{i_1} \approx a^*_{i_2}
	} \sigma_1
	\\
	&=
	\ifelsedef{\top}{i_1 = i_2}{\bot}	
	\end{align*}
	
	And we get that $\A' \sat \tau_1(\xi) \sigma' \sigma_1 \sigma_2 \iff i_1 = i_2 \iff \A \sat \xi \sigma$.

	\proofstep{$\xi = t_1 \approx t_2$ with free variables $x_1, \dots, x_r$, where $t_1,t_2$ are of sort \Nat{}}
	Similarly to the case above, we define $\sigma = \subsi$ and also \linebreak $\sigma' =  \subsi[][r][x^*][a^*]$. For any assignment $\Delta$, we define
	\[
	\Delta' = \Delta \sigma
	\]
		
	Since $t_1, t_2$ are of sort \Nat{}, $\sigma_1$ has no effect, and from \Cref{lem:congruent-nat-open-terms},
	
	\begin{align*}
	\A' &\sat \tau_1(\xi)
	\sigma' \sigma_1 \sigma_2
	\\
	&\iff \A' \sat 
	\brackets{\tau_0(t_1) \approx \tau_0(t_2)}
	\sigma' \sigma_2
	\\
	&\iff 
	\interp'(
	\tau_0(t_1)
	\sigma' \sigma_2
	)
	=		
	\interp'(
	\tau_0(t_2)
	\sigma' \sigma_2
	)
	\\
	\tag{\Cref{lem:congruent-nat-open-terms}}
	&\iff \interp_{\Delta'}(t_1 \sigma) = \interp_{\Delta'}(t_2 \sigma)
	\\
	&\iff \A \sat \xi \sigma
	\end{align*}
	
	\proofstep{$\xi = \neg \zeta$ without free variables}
	Follows from the induction hypothesis for $\zeta$:
	\begin{align*}
	\A' \sat \tau(\xi)
	&\iff \A' \sat \neg \tau(\zeta) \\
	&\iff \A' \nsat \tau(\zeta) \\
	&\iff \A \nsat \zeta \\
	&\iff \A \sat \neg \zeta \\
	&\iff \A \sat \xi
	\end{align*}
	
	\proofstep{$\xi = \neg \zeta$ with free variables $x_1, \dots, x_r$}
	Follows from the induction hypothesis for $\zeta$, for any $i_1, \dots, i_r \in [1,z]$:
	\begin{align*}
	\A' &\sat \tau_1(\neg \zeta) 
	\subsi[][r][x^*][a^*]
	\sigma_1 \sigma_2
	\\
	&\iff \A' \sat \neg \tau_1(\zeta)
	\subsi[][r][x^*][a^*]
	\sigma_1 \sigma_2
	\\
	&\iff \A' \nsat \tau_1(\zeta)
	\subsi[][r][x^*][a^*]
	\sigma_1 \sigma_2
	\\
	\tag{Induction hypothesis}
	&\iff \A \nsat \zeta
	\subsi
	\\
	&\iff \A \sat \neg \zeta
	\subsi
	\\
	&\iff \A \sat \xi 
	\subsi
	\end{align*}
	
	\proofstep{$\xi = \zeta_1 \lor \zeta_2$ without free variables}
	Follows from the induction hypothesis for $\zeta_1$ and $\zeta_2$:
	\begin{align*}
	\A' \sat \tau(\xi) 
	&\iff \A' \sat \tau(\zeta_1) \lor \tau(\zeta_2) \\
	&\iff \A' \sat \tau(\zeta_1) 
	\text{ or } 
	\A' \sat \tau(\zeta_2) \\
	&\iff \A \sat \zeta_1 \text{ or } \A \sat \zeta_2 \\
	&\iff \A \sat \zeta_1 \lor \zeta_2 \\
	&\iff \A \sat \xi	
	\end{align*}
	
	\proofstep{$\xi = \zeta_1 \lor \zeta_2$ with free variables $x_1, \dots, x_r$}
	Follows from the induction hypothesis for $\zeta_1$ and $\zeta_2$, similar to the no free variables case above, since $\tau_1 (\zeta_1 \lor \zeta_2) = \tau_1(\zeta_1) \lor \tau_1(\zeta_2)$.
	
	\proofstep{$\xi = \forall x.\zeta$ without free variables}
	Since $a^{\A'}_i = 0 \iff i > z$ we get the following:
	\begin{align*}
	\A' \sat \tau(\xi) 
	&\iff \A' \sat \bigland_{i = 1}^{\tilde \kappa} 
	\parens{
		a_i \approx 0
		\lor
		\tau_1(\zeta) \sub[x^*][a^*_i]
	} \sigma_1 \sigma_2 
	\\
	&\iff \A' \sat \bigland_{i = 1}^z
	\tau_1(\zeta) \sub[x^*][a^*_i] \sigma_1 \sigma_2
	\\
	&\iff \A' \sat \tau_1(\zeta) \sub[x^*][a^*_i] \sigma_1 \sigma_2 \text{ for all } i \in [1,z]
	\\
	\tag{Induction hypothesis for $\zeta$}
	&\iff \A \sat \zeta \sub[x][\alpha_i] \text{ for all } i \in [1,z]
	\\
	\tag{The set $A$ is covered by $\alpha_1, \dots, \alpha_z$}
	&\iff \A \sat \forall x.\zeta = \xi
	\end{align*}

	\proofstep{$\xi = \forall x.\zeta$ with free variables $x_1,\dots,x_r$}
	Similar to the case above, using the induction hypothesis for $\zeta$ as a formula with free variables $x,x_1,\dots,x_r$.
	
\end{proof}

\proofpart{
	Proof of \Cref{thm:presburger-reduction} ($\To$): 
	If $\varphi$ has a \emph{distinct} \SL{} model with \emph{small \Addr{} space}, 
	then $\varphi'$ has a Standard Model of Arithmetic
}
Let there be a \emph{distinct} \SL{} model for $\varphi$ with \emph{small \Addr{} space}: 
\[
\A = \parens{
	A,
	a^\A_1, \dots, a^\A_l,
	b^\A_1, \dots, b^\A_m,
	c^\A_1, \dots, c^\A_n,
	s^\A_1, \dots, s^\A_m
}
.
\]

We can represent its Addresses set as
$
A = \braces{ \alpha_1, \dots, \alpha_z }
$ where $z = \abs{A}$, and for every $i \in [1,l]$, $a^\A_i = \alpha_i$ since $\A$ is distinct. 
Combined with the fact that $\A$ has small \Addr{} space 
we know that $z \leq \tilde \kappa$.

We define $\A'$, 
the Standard Model of Arithmetic for $\varphi'$, 
as follows:
\[
\A' = \parens{
	a^{\A'}_1, \dots, a^{\A'}_{\tilde \kappa},
	b^{\A'}_{1,1}, \dots, b^{\A'}_{\tilde \kappa,m},
	c^{\A'}_1, \dots, c^{\A'}_n
}
\]
where the indicators are
\[
a^{\A'}_i = \ifelsedef{1}{i \leq z}{0};
\]
the balances are \[
b^{\A'}_{i,j} = \ifelsedef{b^\A_j(\alpha_i)}{i \leq z}{0};
\]
and the natural constants are $c^{\A'}_k = c^\A_k$.

\begin{claim}
	\label{claim:A'-sats-aux}
	The structure $\A'$ satisfies $\eta(\varphi)$.
\end{claim}
\begin{proof}
	We show that $\A'$ satisfies $\eta_1(\varphi)$, $\eta_2(\varphi)$ and $\eta_3(\varphi)$:
	
	\proofstep{$\A'$ satisfies $\eta_1(\varphi)$}
	We need to show that for each $i \in [1, \tilde \kappa]$, 
	\[
	\A' \sat 
	\brackets{
		\parens{ a_i \approx 0 }
		\to 
		\parens{ \bigland_{j=1}^m b_{i,j} \approx 0 }
	}
	.\] I.e. for each $i \in [1, \tilde \kappa]$ and $j \in [1,m]$, if $a^{\A'}_i = 0$, then $b^{\A'}_{i,j} = 0$.
	
	By definition, $a^{\A'}_i = 0 \iff i > z$, in which case, for any $j \in [1,m]$, $b^{\A'}_{i,j} = 0$, as required.
	
	\proofstep{$\A'$ satisfies $\eta_2(\varphi)$}
	We need to show that for each $i \in [1,l]$,
	\[
	\A' \sat a_i \not \approx 0,
	\]
	i.e. $a^{\A'}_i \neq 0$. 
	
	Since $\A$ is a distinct \SL{} model, it has at least $l$ addresses: $l \leq z$. 
	By definition, for any 
	$i \in [1,z]$, $a^{\A'}_i = 1 > 0$, 
	in particular for any 
	$i \in [1,l] \subseteq [1,z]$.
	
	\proofstep{$\A'$ satisfies $\eta_3(\varphi)$}
	We need to show that for each $i \in [1,\tilde \kappa]$, if $a^{\A'}_i = 0$, then for any $i' > i$, $a^{\A'}_{i'} = 0$.
	
	Let there be some index $i$ such that $a^{\A'}_i = 0$, therefore, by definition, $i > z$. For any $i' > i$ it also holds that $i' > z$ and therefore $a^{\A'}_{i'} = 0$.	
\end{proof}

\begin{claim}
	\label{claim:A'-sats-trans}
	The structure $\A'$ satisfies $\tau(\varphi)$.
\end{claim}

\begin{proof}
	We show that $\A$ and $\A'$ are congruent, and since $\A \sat \varphi$, from \Cref{lem:congruent-sat}, $\A' \sat \tau(\varphi)$:
	
	\begin{enumerate}
		\item $\A$ is an \SL{} model of $\varphi$, therefore it holds the sum property.
		
		\item $\A'$ satisfies $\eta(\varphi)$ from \Cref{claim:A'-sats-aux}.
		
		\item $z = \abs{A} \leq \tilde \kappa$ as explained above.
		
		\item For any $i \in [1,l]$, $a^\A_i = \alpha_i$ by definition.
		
		\item By construction of $\A'$, for any $j \in [1,m], i \in [1,z]$, $b^{\A'}_{i,j} = b^{\A}_j(\alpha_i)$ and for any $i > z$, $b^{\A'}_{i,j} = 0$.
		
		\item By construction of $\A'$, for any $i \in [1,z]$, $a^{\A'}_i = 1 > 0$, and for any $i > z$, $a^{\A'}_i = 0$.
		
		\item $\A$ is given to be distinct.
	\end{enumerate}
\end{proof}

\begin{corollary}
	The structure $\A'$ is a Standard Model of Arithmetic for 
	$\varphi' = \tau(\varphi) \land \eta(\varphi)$.
\end{corollary}

\proofpart{
	Proof of \Cref{thm:presburger-reduction} ($\Leftarrow$): 
	If $\varphi'$ has a Standard Model of Arithmetic, 
	then $\varphi$ has a \emph{distinct} \SL{} model with \emph{small \Addr{} space}
}
Let 
$
\A' = \parens{
	a^{\A'}_1, \dots, a^{\A'}_{\tilde \kappa},
	b^{\A'}_{1,1}, \dots, b^{\A'}_{\tilde \kappa, m},
	c^{\A'}_1, \dots, c^{\A'}_n
}
$ 
be a Standard Model of Arithmetic for $\varphi'$.
 Since $\A' \sat \eta_3(\varphi)$, 
 we know that there exists some maximal index 
 $z \leq \tilde \kappa$ such that 
 $a^{\A'}_z \neq 0$ 
 and for any $i > z$, $a^{\A'}_i = 0$. 
 Since $\A' \sat \eta_1(\varphi)$ we know that $z \geq l$.

We construct an \SL{} model $\A$ for $\varphi$ as follows:
\[
\A = \parens{
	A,
	a^\A_1, \dots, a^\A_l,
	b^\A_1, \dots, b^\A_m,
	c^\A_1, \dots, c^\A_n,
	s^\A_1, \dots, s^\A_m
}
\]
where the Addresses set is
\[
A = [1,z];
\]
the Address constants are
\[
a^\A_i = i
\] for any $i \in [1,l]$; the balances are
\[
b^\A_j(i) = b^{\A'}_{i,j}
\] for any $i \in A, j \in [1,m]$; the natural constants are $c^\A_k = c^{\A'}_k$ and the sums are defined as
\[
s^\A_j = \sum_{\alpha \in A} b^\A_j(\alpha).
\]

~\\ \noindent
We show that $\A,\A'$ are congruent:
\begin{enumerate}
	\item By construction, $\A$ holds the sum property.
	
	\item It is given that $\A'$ satisfies $\eta(\varphi)$.
	
	\item We define $A$ to be the set $[1,z]$, and therefore $\abs{A} \leq \tilde \kappa$.
	
	\item $a^\A_i = \alpha_i$ as defined above.
	
	\item By construction, for any $j \in [1.m], i \in [1,z]$, $b^\A_j(\alpha_i) = b^{\A'}_{i,j}$ and $z$ was chosen such that for any $i > z$, $b^{\A'}_{i,j} = 0$.
	
	\item $z$ was chosen such that for any $i \in [1,z]$, $a^{\A'}_i > 0$ and for any $i > z$, $a^{\A'}_i = 0$.
	
	\item $\A$ is distinct by construction.
\end{enumerate}

Given that $\A' \sat \varphi'$, 
we know in particular that 
$\A' \sat \tau(\varphi)$, 
and from \Cref{lem:congruent-sat}, 
$\A \sat \varphi$ as a many-sorted, first-order formula. 
Since $\A$ holds the sum property, 
$\A \sat_\text{SL} \varphi$. 
In addition, by construction 
$\abs{A} = z \leq \tilde \kappa$. 
Therefore, $\A$ is a \emph{distinct} \SL{} model for $\varphi$ with \emph{small \Addr{} space}.
\end{longproof}

\newpage
\section{Proofs for Encodings}
\subsection{Soundness of $\mint_1$}
 For readability reasons, the theorem is stated again.
 \soundMint*

Now the proof is as follows.
\begin{proof}
 To show that $\newsum=\oldsum+1$, we consider the subformula (M3) of $\mint_1(a,c)$. It follows that $\oldact\backslash\{c\}=\newact\backslash\{c\}$. Now using (M1), we get $\oldact=\oldact\backslash\{c\}$ and $(\newact\backslash\{c\}) \cup \{c\}=\newact$. Thus, $\oldact\;\dot{\cup}\;\{c\}=\newact$ and hence \linebreak $|\oldact|+1=|\oldact \;\dot{\cup} \; \{c\}|=|\newact|$ which implies \linebreak $\newsum=\oldsum+1$.
 
 Similar reasoning works to show  $\newbal(a)= \oldbal(a)+1$. From (M4) it follows \[\{d\mid (a,d)\in \oldhc\}\backslash\{c\}=\{d\mid (a,d)\in \newhc\}\backslash\{c\}\;.\] Now using (M2) we get
  \begin{align*} 
 & \{d\mid (a,d)\in \oldhc\}=\{d\mid (a,d)\in \oldhc\}\backslash\{c\} \quad \text{and} \\
 & \{d\mid (a,d)\in \newhc\}\backslash\{c\} \cup \{c\}=\{d\mid (a,d)\in \newhc\} \;.
 \end{align*}
 As before, we have \[\{d\mid (a,d)\in \oldhc\} \,\dot{\cup}\, \{c\}=\{d\mid (a,d)\in \newhc\}\;,\] which implies $\newbal(a)= \oldbal(a)+1$.
 
 Finally, $\newbal(b)=\oldbal(b)$ for $b \neq a $ follows from (M4), since it implies $\{d\mid (b,d)\in \oldhc\}=\{d\mid (b,d)\in \newhc\}$.\qed
\end{proof}

\subsection{Soundness of $\mint_n$} \label{apd:mint_n}  To define the smart transition $\mint_n$ we need one pair of predicates for every time step. Thus we have an additional "parameter" $i$, the $i$-th time step, in \ac{} and \hc{} instead of using the prefixes \old{-} and \new{-}. Other than that the definition and the soundness result is analog to the setting of $\mint_1$.
	
	\begin{definition}[Transition $\mint_n(a)$]\label{mintN}
		Let  $a \in \Addr$. 
		Then, the transition $\mint_n(a)$ activates $n$ coins 
		and deposits them into address $a$, one coin $c$ in each time step.
		\begin{enumerate}
			\item 	
			The coin $c$ was inactive before and is active now: 
			\[ 
			\tag{N1} 
			\neg \ac(c,i) \land \ac(c,i+1) \;.
			\]
			
			\item 
			The address $a$ owns the new coin $c$:
			\[ 
			\tag{N2}  
			\hc(a,c,i+1) \land \forall a' : \Addr{}.\; \neg \hc(a',c,i)\;.
			\]
			
			\item 
			Everything else stays the same: 
			\begin{align*}
				\tag{N3}
				& \forall  c' : \Coin.\;\forall a' : \Addr. \;
				\; 
				(c' \not\approx c \lor a' \not\approx a) \to \\
				& \begin{array}{ll}
					\quad  \big( 
					&\parens{ \ac(c',i+1) \liff \ac(c',i) } \land \\
					&   \parens{ \hc(a',c',i+1) \liff \hc(a',c',i+1) } 
					\; \big)\;.
				\end{array}
			\end{align*}
		\end{enumerate}
		The transition $\mint_n(a)$ is defined as 
		$\forall i: \Nat.\; \exists c :\Coin.\; \text{(N1)} \land \text{(N2)} \land \text{(N3)}$.
	\end{definition}
	The soundness result we get is similar to \Cref{mint1_sound} but extended by the new parameter.
	
	\begin{theorem}[Soundness of $\mint_n(a)$]\label{mintn_sound}
		Let $a\in \Addr$ such that 
		$\mint_n(a)$. 
		Consider a balance function
		$\bal: \Addr\times\N \to \N$,
		a summation function $\msum: \N \to \N^+$,
		a binary predicate $\ac \subseteq \Coin\times\N$
		and a ternary predicate $\hc \subseteq \Addr \times \Coin\times \N$
		such that for every $i \in \N$
		\[
		\abs{\ac(.,i)}=\msum(i)\, 
		\] 
		and for every address $a'$ and $i \in \N$, we have 
		\begin{align*}
			&\bal(a',i)=\abs{ \braces{ c' \in \Coin \mid (a', c',i) \in \hc  }}\;.
		\end{align*}
		
		Then for an arbitrary $n \in \Nat$,
		$\msum(n) = \msum(0) + n$,
		$\bal(a,n)= \bal(a,0)+n$. Moreover, for all other addresses $a' \neq a$, 
		it holds $\bal(a',n)=\bal(a',0)$.
	\end{theorem}
\begin{proof}
We prove \Cref{mintn_sound} by induction over $n \in  \N$. The base case $n=0$ is trivially satisfied. For the induction step, we get the induction hypothesis $\msum(n) = \msum(0) + n$,
$\bal(a,n)= \bal(a,0)+n$, $\forall a' \neq a.\;\bal(a',n)=\bal(a',0)$. By defining $\oldsum\define \msum(n)$, $\newsum\define\msum(n+1)$ and analogously for \ac{}, \bal{} and \hc{}, all the preconditions of \Cref{mint1_sound} hold. Therefore, we get $\msum(n+1)=\msum(n)+1$, $\bal(a,n+1)= \bal(a,n)+1$, $\forall a'\neq a.\; \bal(a',n+1)=\bal(a,n)$, by applying \Cref{mint1_sound}. Together with the induction hypothesis this yields  $\msum(n+1) = \msum(0) + n+1$,
$\bal(a,n+1)= \bal(a,0)+n+1$, $\forall a' \neq a.\;\bal(a',n+1)=\bal(a',0)$ and thus concludes the induction proof.
\qed	
\end{proof}

\subsection{Soundness and Completeness relative to $f$} \label{app:f}
In order to establish a proof of \Cref{sound_comp}, some formal definitions of the in the paper informally explained concepts have to be stated first.
The exclusion of certain elements of $\mF$ is based on an equivalence relation $\sim$.

We first formally define $\sim$:
\begin{definition}[Relation $\sim\,$]
	Let the pairs 
	$p_1=(\hco,\aco)\in\mF$ and 
	$p_2=(\hct,\act) \in \mF$. Then 
	$p_1 \sim p_2$ iff
	\begin{enumerate}
		\item 
		$|\aco|=|\act| \;$,
		\item
		$
		\abs{\braces{c \in \Coin{}\mid \hco(a,c) }}
		=
		\abs{\braces{c \in \Coin{}\mid \hct(a,c) }}
		$, for all $a\in \Addr$,
		\item $p_1$ violates (I2) in $V_{\leq}$ cases and
		$p_2$ violates (I2) also $V_{\leq}$ times; 
		\item $p_1$ does not satisfy (I1) and (I3) in all
		together 	$V_{\geq}$ cases, which is also the
		number of times $p_2$ violates (I1) and (I3). 
	\end{enumerate}
\end{definition}

To properly prove that $\sim$ is an equivalence relation, we have to define $V_{\leq}$ and $V_{\geq}$ first.

		\begin{definition}
		Given a pair $(\ac,\hc) \in \mF$. For an address $a$, we define
		$C_a \define \{ c \in \atexttt{Coin}\mid \hc(a,c)\}$. Further, we define three types of error coins:
		\begin{enumerate}
			\item $M_{\text{Inact}} \define \{ c \in \atexttt{Coin}\mid \neg\ac(c) \land \exists a.\; c \in C_a \} $,
			\item $M_{\text{Least}} \define \{c \in \atexttt{Coin}\mid \ac(c) \land \forall a.\; c \notin C_a  \}$ and
			\item $M_{\text{Most}} \define \{ c \in \atexttt{Coin}\mid \,\exists a,b.\; a\not\approx b \land c \in C_a \land c \in C_b \} $
		\end{enumerate}
		and one type of error pairs $ M_{\text{Pairs}}\define\{(a,c)\mid  c \in C_a \land \exists b.\; a \not\approx b \land c \in C_b \}$
		to refine the number of mistakes caused by the violation of (I3).\\
		The number of violations of (I2) is now $V_{\leq}\define |M_\text{Least}|$.
		and the  number of violations of (I1) and (I3) is defined as $V_{\geq}\define|M_{\text{Inact}}|+|M_{\text{Pairs}}|-|M_{\text{Most}}|$.
	\end{definition}

	\begin{lemma}
	The relation $\sim$ is an equivalence relation on $\mF$.
\end{lemma} 

\begin{proof}
	\begin{itemize}
		\item Reflexivity of $\sim$.\\
		Let $(\hc,\ac) \in \mF$, then clearly $|\ac|=|\ac|$, for all $a$ we have $|C_a|=|C_a|$ and also $V_{\leq}=V_{\leq}$, $V_{\geq}=V_{\geq}$. Hence $(\hc,\ac)\sim(\hc,\ac)$.
		
		\item Symmetry of $\sim$.\\
		Let $p_1, p_2 \in \mF$ such that $p_1 \sim p_2$, then due to symmetry of $=$ also $p_2 \sim p_1$ holds.
		
		\item Transitivity of $\sim$.\\
		Let $p_1,p_2,p_3\in \mF$, such that $p_1 \sim p_2$ and $p_2 \sim p_3$ then due to the transitivity of $=$ also $p_1 \sim p_3$ holds. \qed
	\end{itemize}
\end{proof}

The translation function $f$ can now be defined as a function that
assigns every pair $(\bal,\msum)$ a class from $\mF/_{\sim}$.

\begin{definition}[Translation Function $f$]\label{def:f_welldef}  
	The function 
	$f: \N^{\Addr} \times \N \to \mF/_{\sim}$, 
	$(\bal,\msum) \mapsto [(\hc,\ac)]_{\sim}$, 
	is defined to satisfy the following conditions for an arbitrary
	$(\hc,\ac) \in [(\hc,\ac)]_{\sim}$.
	\begin{enumerate}
		\item 
			$\msum = \abs{ \ac }\;.$
		\item 
			For every $a \in \Addr$
			it holds 
			$\bal(a)= \abs{ \braces{ c \in \Coin \mid \hc(a,c) } }$. 
		\item 
			At least one of 
			$V_{\leq} = 0$ 
			and 
			$V_{\geq} = 0$ 
			holds.
	\end{enumerate}
\end{definition}

The function $f$ is well-defined and injective,
ensuring soundness and completeness of our \SL{} encodings relative to $f$.

\completeImplicitSL*

\begin{proof}
	The proof is organized in 4 steps. The first step provides a technicality that is need for the steps 2 and 3 and finally in the last step the claim is proven.
	\begin{enumerate}
		\item Consider any pair $(\hc,\ac) \in \mF$ with $V_{\leq}= V_{\geq}=0$. Then, since there are no coins nor addresses violating the invariants here, we thus have  $\bigcup_{a \in \texttt{Address}} C_a = \ac$ and all the $C_a$ are disjoint. Thus, $\sum_{a \in \texttt{Address}} |C_a|= |\bigcup_{a \in \texttt{Address}} C_a |=|\ac|$.
		\item Now  we only assume $V_{\geq}=0$. Consider \[M_{\text{Least}} = \{c \in \texttt{Coin}\mid \ac(c) \land \forall a.\; c \notin C_a  \} \subseteq \ac\;. \] Then the pair $p'=(\hc, \ac\backslash M_{\text{Least}})$ satisfies $V'_{\leq}=V'_{\geq}=0$, because all the coins were not active originally are active now and we did not change the any of the other mistake sets. From the first step we now get \[\sum_{a \in \texttt{Address}} |C'_a| =|\ac\backslash M_{\text{Least}}|\] and therefore \[ \sum_{a \in \texttt{Address}} |C_a| =|\ac|+V_{\geq}-V_{\leq}\;.\]
		\item Similarly to the second step we now only assume $V_{\leq}=0$. By definition it holds that $M_{\text{Inact}}\cap \ac = \emptyset$, $M_{\text{Pairs}}\subseteq \hc$ and $ M_{\text{Most}} \subseteq \ac \cup M_{\text{Inact}} $. We now consider the pair \[p''=(\hc\backslash M_{\text{Pairs}}, (\ac \cup M_{\text{Inact}})\backslash M_{\text{Most}}  )\;. \] Clearly, there is not any coin assigned to two different addresses in $p''$. However all the coins that were in two different addresses before are now not assigned to any address, this is why these coins have to be removed from $\ac \cup M_{\text{Inact}}$. Also there are no coins that are active without belonging to any address. Further, all active coins still are assigned to an address as the problematic ones have been removed.
		Hence, $V''_{\leq}=V''_{\geq}=0$. Now, we can again apply the result of the first step to get \[\sum_{a \in \texttt{Address}} |C_a|-|M_{\text{Pairs}}| =|\ac| + |M_{\text{Inact}}| - |M_{\text{Most}}|\] and thus \[ \sum_{a \in \texttt{Address}} |C_a| =|\ac|+V_{\geq}-V_{\leq}\;.\]
		
		\item Using the results of the previous two steps we can now prove the theorem. Let $(\bal, \msum) \in \N^{\texttt{Address}}\times\N$ and $(\hc,\ac) \in f(\bal,\texttt{sum})$, then $V_{\leq}=0$ or $V_{\geq}=0$. In both cases it follows $\sum_{a \in \texttt{Address}} |C_a| =|\ac|+V_{\geq}-V_{\leq}$ and therefore by definition of $f$ it holds $\sum_{a \in \texttt{Address}} \bal(a) =\texttt{sum} +V_{\geq}-V_{\leq}$.
		Assume now $\sum_{a \in \texttt{Address}} \bal(a) =\texttt{sum} $. It follows $V_{\geq}-V_{\leq}=0$ and since we know that one of these values has to be zero by definition of $f$ it holds $V_{\leq}= V_{\geq}=0$. But this statement is equivalent to $\texttt{inv}\,(\hc{},\;\ac{})$. For the other direction assume $\texttt{inv}\,(\hc{},\;\ac{})$, this implies $V_{\leq}= V_{\geq}=0$ and hence $\sum_{a \in \texttt{Address}} \bal(A) =\texttt{sum} $. This conlcudes the  proof.
	\end{enumerate}
	 \qed
\end{proof}

From \Cref{sound_comp}, also \Cref{thm:enc:soundness} follows immediately by stating the properties of the function $f$.

\subsection{Soundness of Explicit \SL{} Encodings} 
\soundExplicitSL*

\begin{proof}
 Consider 	$(\bal,\texttt{sum})$, $(\hc,\ac)$ as in the theorem. 
 Then by property (Ax1) and the codomain of $\mcount$ we have $\ac = \{c \in \texttt{Coin}\mid \mcount(c) \in [1,\texttt{sum}] \}$.
  Since $\mcount$ is bijective, it holds \[|\ac|=| \{c \in \texttt{Coin}\mid \mcount(c) \in [1,\texttt{sum}] \}|=\texttt{sum}\;.\] Similarly, by (Ax2) and the codomain of $\ind(a,.)$ we know $C_a = \{c \in \texttt{Coin}\mid \ind(a,c) \in [1,\bal(a)] \}$. As $\ind(a,.)$ is bijective as well it follows \[|C_a|=| \{c \in \texttt{Coin}\mid \ind(a,c) \in [1,\bal(a)] \}|=\bal(a)\;.\]
 Hence  $(\hc,\ac) \in f(\bal,\texttt{sum})$ and by \Cref{sound_comp}, we get \linebreak $\sum_{a \in \Addr} \bal(a)=\texttt{sum}$, since  $\texttt{inv}\,(\hc{},\;\ac{})$. \qed
\end{proof}

\subsection{No Loss of Generality: Restricting \ind{} and \mcount}\label{app:onecount}
We want to prove that we do not lose any generality when considering mutual \mcount{} and \ind{} functions for the old- and the new-world. In order to do so we need the following preliminary lemmas. 

\begin{lemma} \label{element_exists}
	Given two pairs $h_x\define(\xbal,\xsum)$, $h_y\define(\ybal,\ysum)$ with $\sum_{a \in \Addr} \zbal(A)=\zsum$, for $\atexttt{z} \in \{\atexttt{old},\atexttt{new}\}$  and $\xsum \leq \ysum$. Further, let $p_x=(\xhascoin,\xactive)\in f(h_x)$.\\
	Then there exists $p_y=(\yhascoin,\yactive) \in f(h_y)$ satisfying the following properties:
	\begin{enumerate}
		\item $\xactive \subseteq \yactive$.
		\item $\xbal(a) \leq \ybal(a) \; \Rightarrow\; C_{x,a} \subseteq C_{y,a}$.
		\item $\ybal(a) \leq \xbal(a) \; \Rightarrow\; C_{y,a} \subseteq C_{x,a}$.
	\end{enumerate}
\end{lemma}
\begin{proof} We proceed by constructing $p_y=(\yhascoin,\yactive) \in f(h_y)$ such that it satisfies properties (1)-(3).
	To fulfill property (1), let $\yactive\define$ \linebreak $\xactive \cup S$, where $S \in \texttt{Coin}\backslash \xactive$ and $|S|=\ysum - \xsum$. Then also $|\yactive|=\ysum$ holds. To construct the $C_{y,a}$ properly, the set $\yactive$ has to be partitioned, since $p_y \in f(h_y)$ and thus $\texttt{inv}(\yhascoin, \yactive)$. For every $a$ with $\xbal(a) \leq \ybal(a)$ we require $C_{x,a} \subseteq C_{y,a}$. Therefore there are \[\sum_{a:\; \xbal(a)>\ybal(a)} \xbal(a)-\ybal(a) \] additional spare coins. For $a$ with  $\xbal(a) \geq \ybal(a)$ we want $C_{y,a} \subseteq C_{x,a}$, which leaves us with \[\sum_{a:\; \xbal(a)<\ybal(a)} \ybal(a)-\xbal(a) \] missing coins. Hence, the difference is 
	\begin{align*}
		&\;& \ysum-\xsum& &+\quad& \sum\limits_{\substack{a \in \Addr, \\ \xbal(a)>\ybal(a) }} \xbal(a)-\ybal(a)\\
		&\;& \quad& &-\quad& \sum\limits_{\substack{a \in \Addr, \\ \ybal(a)>\xbal(a) }} \ybal(a)-\xbal(a) \quad .
	\end{align*}
By replacing $\zsum$ by $\sum_{a \in \Addr} \zbal(a)$ all the summands with either \linebreak $\ybal(a)>\xbal(a)$ or $\xbal(a)>\ybal(a)$ disappear and the remaining value is 0. Therefore, such a partition of $\yactive$ exists and thus, there exists $p_y=(\yactive,\yhascoin) \in f(h_y)$ satisfying (1),  (2) and (3). \qed
\end{proof}

\begin{lemma}\label{step1}
	Given two pairs $h_x, h_y$ with  $\sum_{a \in \Addr} \zbal(a)=\zsum$, $p_z \in f(h_z)$, for $\atexttt{z} \in \{\atexttt{x},\atexttt{y}\}$  and $\xsum \leq \ysum$ as in Lemma \ref{element_exists}. Then, there exist a bijective function $\mcount: \atexttt{Coin} \to \N^+$ with $\mcount(\zact)=[1,\zsum]$ and bijective functions $\ind(a,.): \atexttt{Coin} \rightarrow \N^+$, with $\ind(C_{z,a})=[1,\zbal(a)]$, for $\atexttt{z} \in \{\atexttt{x},\atexttt{y}\}$, $a \in \Addr$.
\end{lemma}

\begin{proof}
	At first, we construct \mcount{}. We know $\yactive = \xactive\, \dot{\cup} \,S$, where $|\yactive|=\ysum$, $|\xactive|=\xsum$ and $|S|=\ysum-\xsum$. Thus, we can easily find an injective function with $\mcount(\xactive)=[1,\xsum]$ and $\mcount(S)=[\xsum+1,\ysum]$. Further, this function can be bijectively extended onto $\N^+$. Similarly, for the addresses $a$, we construct $\ind(a,.)$ in the following way. Since we know $|C_{z,a}|=\zbal(a)$, we can find an injective function with $\ind(a,C_{x,a})=[1,\xbal(a)]$. For all $a$, where $\ybal(a) \leq \xbal(a)$, we can assume that $\ind(a, C_{y,a})=[1, \ybal(a)]$, as $C_{y,a} \subseteq C_{x,a}$. For these addresses $a$, $\ind(a,.)$ can now also be extended bijectively onto $\N^+$. Finally, for $a$ with $\xbal(a) \leq \ybal(a)$ we know $C_{x,a} \subseteq C_{y,a}$ and can thus assume $\ind(a,C_{y,a}\backslash C_{x,a})=$\linebreak$[\xbal(a)+1, \ybal(a)]$. Now also these $\ind(a,.)$ can be extended bijectively onto $\N^+$. \qed
\end{proof}

Having these two lemmas at hand we can now state and prove the following result.

\begin{theorem}
	Given any two pairs $h_o\define(\oldbal,\oldsum)$,  \linebreak $h_n\define(\newbal,\newsum )$ with $\sum_{a \in \atexttt{Address}} \zbal(a)=\zsum$, \linebreak for $\atexttt{z} \in \{\atexttt{old},\atexttt{new}\}$. There exist bijective functions $\mcount: \atexttt{Coin} \rightarrow \N^+$ and $\ind(a,.): \atexttt{Coin} \rightarrow \N^+$, for every $a \in \atexttt{Address}$  such that there are
	$p_o=$\linebreak $(\oldact, \oldhc) \in f(h_o)$, $p_n=(\newact, \newhc)$\linebreak $\in f(h_n)$ with
	\begin{align} \label{fml1}
		\forall &c.\; \big(\zact(c) \leftrightarrow \mcount(c) \leq \zsum\big) \quad \text{and} \\
		 \label{fml2} \forall &a, c.\; \big(\zhc(a,c) \leftrightarrow \ind(a,c) \leq \zbal(a)\big)  \; .
	\end{align}
\end{theorem}

\begin{proof}
	Let $h_x \in \{h_o, h_n\}$ such that $ \xsum = \min \left\{\oldsum, \newsum \right\}$. The other pair gets the prefix '\texttt{y-}' from now on. Also elements in $f(h_x)$ and $f(h_y)$ will be named accordingly. 
	
	\noindent Let $(\xactive,\xhascoin) \in f(h_x)$ arbitrary, $(\yhascoin,\yactive) \in f(h_y)$ as in Lemma \ref{element_exists} and $\mcount$, $\ind$ as in Lemma \ref{step1}. 
	
	\noindent Then it holds $\mcount(\zact)=[1,\zsum]$. Thus, $\forall c.\; c \in \zact \rightarrow \mcount(c) \in [1,\zsum]$. As \mcount{} is bijective and therefore injective, it follows $\forall c.\; c \notin \zact \rightarrow \mcount(c) \notin [1,\zsum]$. Together with the fact that the codomain of \mcount{} is $\N^+$ we get Formula \ref{fml1}.
	The analog argumentation works for Formula \ref{fml2}. We know $\ind(a,C_{z,a})=[1, \zbal(a)]$. Thus  $\forall c.\; c \in C_{z,a} \rightarrow \ind(a,c) \in [1,\zbal(a)]$. Also $\ind(a,.)$ is bijective and therefore injective which implies $\forall c.\; c \notin C_{z,a} \rightarrow \ind(a,c) \notin [1,\zbal(a)]$. By definition of $C_{z,a}$ and the codomain of $\ind(a,.)$ Formula \ref{fml2} holds. This concludes the proof. \qed
\end{proof}

\subsection{No Loss of Generality: Ordering of Coins} \label{app:ordering}
The property to prove is that whenever a block of coins has the same order in two of our counting functions and they are not crossing its crucial value (\texttt{sum}, \bal($a_i$) ), then we can assume that they are ordered in the same way.
In order to do so, we have to formalize the notion of the former invariants $\texttt{inv'}(\ind,\mcount)$. They are the formulas one gets by replacing \hc{} and \ac{} by \mcount{} and \ind{} according to (Ax1) and (Ax2) in the invariants (I1)-(I3).

\begin{definition} Let $\mcount{}: \atexttt{Coin} \rightarrow \N^+$ and $\ind: \atexttt{Address}\times\atexttt{Coin} \rightarrow \N^+$, then with the formulas
	\begin{align*}
	\tag{I1'} &\forall c.\; \; (\exists a.\; \ind(a,c) \leq \bal(a)) \leftrightarrow \mcount(c) \leq \atexttt{sum}\;, \\ 
	\tag{I2'} & \forall a, b, c.\; \; (\ind(a,c) \leq \bal(a) \land \ind(b,c) \leq \bal(b)) \rightarrow a \approx b 
	\end{align*}
we define $\texttt{inv'}(\ind,\mcount)\define \text{I1'}\land \text{I2'}$.
\end{definition}

\begin{theorem}
	Let
	\begin{enumerate}
		\item[(i)] $(\oldbal, \oldsum)$, $(\newbal, \newsum) \in \N^{\Addr} \times \N$,
		\item[(ii)] $\mcount: \atexttt{Coins} \rightarrow \N^+$ bijective,
		\item[(iii)]  $\ind: \atexttt{Address}\times\atexttt{Coin} \rightarrow \N^+$, such that $\ind(A,.)$ bijective for every $a$ and
		\item[(iv)]  $\texttt{inv'}(\ind,\mcount)$. 
	\end{enumerate}
	If now 
	\begin{itemize} 
		\item[(v)] $\forall c : \atexttt{Coin}.\; f_0(c) \in [l_0,u_0] \leftrightarrow f_1(c) \in [l_1, u_1]$, where $f_0,f_1 \in \{\mcount\} \cup \{\ind(a,.): a \in  \atexttt{Address}\}$ and
		\item[(vi)]  either  $u_i \leq x_i$ or $x_i < l_i$ for $i \in \{0,1\}$, where
		 \subitem  $x_i \define \bal(a_i)$, if $f_i=\ind(a_i,.)$, or 
		 \subitem$x_i\define \atexttt{sum}$, if $f_i=\mcount$,
	\end{itemize}
	 then there exist $\ind{'}$, $\mcount{'}$ with the properties $(i)$-$(vi)$ and $\forall c :\atexttt{Coin}.\; f_0'(c) \in [l_0,u_0] \rightarrow f_0'(c) +l_1 \approx f_1'(c) + l_0$.
\end{theorem}

\begin{proof}
For simplicity assume $f_0=\ind(a_0,.)$ and $f_1=\ind(a_1,.)$. However, the proof works in an analog way for \mcount{}. We proceed by constructing \ind{'} and \mcount{'} and then showing properties  $(i)$-$(vi)$ and $\forall c : \texttt{Coin}.\; f_0'(c) \in [l_0,u_0] \rightarrow f_0'(c) +l_1 = f_1'(c) + l_0$ hold.\\
The function $\mcount' \define \mcount$. We construct \ind{'} the following way.
\begin{align*}
&\forall a,c.\; \;(a\not\approx a_1 \lor \ind(a_1,c) \notin [l_1,u_1]) \text{ let } \ind'(a,c)\define \ind(a,c) \quad\text{and}\\
& \forall c.\; \ind(a_1,c) \in [l_1,u_1] \text{ we define } \ind'(a_1,c) \define \ind(a_0,c) -l_0 + l_1 \;. 
\end{align*}
With these definitions the properties $(i)$, $(ii)$, $(vi)$ and $\forall c : \texttt{Coin}.\; f_0'(c) \in [l_0,u_0] \rightarrow f_0'(c) +l_1 = f_1'(c) + l_0$ obviously hold.

To show property $(v)$ we first fix a coin $c$ and assume $\ind'(a_0,c) \in [l_0,u_0]$. By definition of \ind{'} we know that also $\ind(a_0,c) \in [l_0,u_0]$ and then using $(v)$ we get $\ind(a_1,c) \in [l_1,u_1]$. Again by definition of \ind{'}, we have $ \ind'(a_1,c) = \ind(a_0,c) -l_0 + l_1$. Since $\ind(a_0,c) \in [l_0,u_0]$, it follows $ \ind'(a_1,c) = \ind(a_0,c) -l_0 + l_1 \in [l_1, u_0-l_0+l_1]$. From property $(v)$ and the bijectivity of the $\ind(a,.)$ functions, it follows that the intervalls have the same size and thus $u_0-l_0=u_1-l_1$. Therefore we end up having $\ind'(a_1,c) \in [l_1,u_1]$.\\
For the other direction of the equivalence, we fix $c$ and assume $\ind'(a_0,c) \notin [l_0,u_0]$. Then $\ind(a_0,c) \notin [l_0, u_0]$ and using $(v)$ we get $\ind(a_1,c) \notin [l_1,u_1]$ and thus \linebreak $\ind'(a_1,c)=  \ind(a_1,c) \notin [l_1,u_1]$. This concludes the proof of $(v)$.\\
Note that $(v)$ implies $\ind'(a_1,c) \in [l_1,u_1]$ if and only if $\ind(a_1,c) \in [l_1,u_1]$.

The proof of property $(iii)$ only requires showing $\ind'(a_1,.)$ is bijective. For showing injectivity, assume 
$\ind'(a_1,c)=\ind'(a_1,d)$. where $c \neq d$. Then clearly \linebreak$\ind(a_1,c) \neq \ind(a_1,d)$. Thus, one of $\ind(a_1,c), \ind(a_1,d) \in [l_1,u_1]$, since otherwise  $\ind'(a_1,c)=\ind(a_1,c)$ and $\ind'(a_1,d)=\ind(a_1,d)$. From the implication of $(v)$, we know that $\ind'(a_1,c)=\ind'(a_1,d) \in [l_1,u_1]$ and thus $\ind(a_1,c)$, $\ind(a_1,d) \in [l_1,u_1]$. Therefore $\ind(a_0,c) -l_0 + l_1 =\ind'(a_1,c)=\ind'(a_1,d) = \ind(a_0,d) -l_0 + l_1$ and thus $\ind(a_0,c)=\ind(a_0,d)$ which is a contradiction to the bijectivity of $\ind$.\\
For surjectivity of $\ind'(a_1,.)$ let $n \in \N^+$ be arbitrary. Assume $n \notin [l_1,u_1]$ first. Then by surjectivity of $\ind(a_1,.)$ it follows that there is a $c$ such that $\ind(a_1,c)=n$ and as $n \notin [l_1,u_1]$, we have $\ind'(a_1,c) =n$. Assume now $n \in [l_1,u_1]$, then by surjectivity of $\ind(a_0,.)$, there exists $c$ such that $\ind(a_0,c) = n-l_1+l_0$. With the same reasoning as above it follows $n-l_1+l_0 \in [l_0,u_0]$ and therefore we conclude $\ind'(a_1,c) = \ind(a_0,c)-l_0+l_1 = n$. This completes the surjectivity proof.

Finally, we have to prove $(iv)$. Once we have shown $\ind(a,c) \leq \bal(a)$ iff $\ind'(a,c) \leq \bal(a)$ for all $a$ and for all $c$ the property follows immediately, since $\mcount'=\mcount$. For all $a$, $c$ with one of $a\neq a_1$ or $\ind(a_1,c) \notin [l_1,u_1]$ the equivalence follows from the definition of $\ind'$. Consider now $a_1$ with a $c$ such that $\ind(a_1,c) \in [l_1,u_1]$. Using the implication of $(v)$, we get $\ind'(a_1,c) \in [l_1,u_1]$. Now using property $(vi)$,  we know that either $u_1 \leq \bal(a_1)$, in which case $\ind(a_1,c)$, $\ind'(a_1,c) \leq \bal(a_1)$, or $\bal(a_1) < l_1$ which implies $\ind(a_1,c)$, $\ind'(a_1,c) > \bal(a_1)$. This concludes the proof of property $(iv)$ and thus of the theorem.\qed
 \end{proof}

\newpage
\section{SMTLIB Encodings and Function Calls} \label[appendix]{app:codes}
In this section the concrete function calls for each table are listed together with one example encoding. 
 \subsection{Codes for Table \ref{tbl:total_surj_table}} 
 The precise function calls used are as follows.
 \begin{itemize}
 	\item For Z3 default:\\  \texttt{z3 -smt2 <file-name>}\\
 	\item For CVC4 default with full-saturate-quant:\\  \texttt{cvc4 --lang=smtlib2.6 --full-saturate-quant <file-name>}\\
 	\item  For Vampire default:\\  \texttt{vampire -input\underline{ }syntax smtlib2 <file-name>}\\
 \end{itemize}

\noindent {\bf Encoding of \texttt{mint1}, full surjectivity, with \texttt{total}, \texttt{int} version:}
 \begin{lstlisting}
(set-logic UFLIA)
(declare-sort Coin 0 )
(declare-sort Address 0)
(declare-fun old-sum () Int)
(declare-fun new-sum () Int)
(declare-fun old-total () Int)
(declare-fun new-total () Int)
(declare-fun a0 () Address)
(declare-fun old-bal (Address) Int)
(declare-fun new-bal (Address) Int)
(declare-fun count (Coin) Int)
(declare-fun ind (Coin Address) Int)

;###  axioms on sum and count  ###
;#sum non-negative
(assert (<= 0 old-sum))
(assert (<= 0 new-sum))
;#count positive
(assert (forall ((C Coin)) (< 0 (count C)) ) )
;#count injective
(assert (forall ((C Coin) (D Coin))
 (=> (= (count C) (count D)) (= C D) )))
;#count surjective
(assert (forall ((N Int))
 (=>
  (and (< 0 N) (or (<= N old-sum) (<= N new-sum)) )
  (exists ((C Coin)) (= (count C) N) ))))

;####  axioms on bal and ind ###
;#bal non-negative
(assert (forall ((A Address)) (<= 0 (old-bal A)) ))
(assert (forall ((A Address)) (<= 0 (new-bal A)) ))
;#ind positive
(assert (forall ((C Coin)(A Address)) (< 0 (ind C A)) ))
;#ind(A,.) injective
(assert (forall ((C Coin) (D Coin) (A Address))
 (=> (= (ind C A) (ind D A)) (= C D) )))
;#ind(A,.) surjective
(assert (forall ((N Int) (A Address))
 (=> (and (< 0 N) (or (<= N (new-bal A)) (<= N (old-bal A))))
	(exists ((C Coin)) (= (ind C A) N) ))))
	
;###  axioms between sum and bal  ###
; #ind leq bal iff count leq sum
(assert (forall ((C Coin)) (=
 (exists ((A Address)) (<= (ind C A) (old-bal A)) )
 (<= (count C) old-sum) )))
(assert (forall ((C Coin)) (=
 (exists ((A Address)) (<= (ind C A) (new-bal A)) )
 (<= (count C) new-sum) )))
;#only once ind leq bal
(assert (forall ((A Address)(B Address)(C Coin))
	(=> (and
         (<= (ind C A) (old-bal A) )
         (<= (ind C B) (old-bal B) ) )
        (= A B) )))
(assert (forall ((A Address)(B Address)(C Coin))
 (=> (and
      (<= (ind C A) (new-bal A) )
      (<= (ind C B) (new-bal B) ) )
     (= A B) )))

;###  transition and expected impact  ###
;#mint1
(assert (and 
   (= (new-bal a0) (+ (old-bal a0) 1))
   (forall ((A Address ))
      (=> (distinct A a0) (= (old-bal A) (new-bal A)) ))))
;#expected Impact
(assert (= (+ old-sum 1) new-sum) )

;### invariants ###
;#pre-invariant
(assert (= old-sum old-total) )
;#negated post-invariant
(assert (distinct new-sum new-total) )

(check-sat)
 \end{lstlisting}

 \subsection{Codes for Table \ref{tbl:overview_table}}
  The precise function calls used are as follows.
 \begin{itemize}
 	\item For Z3 default:\\  \texttt{z3 -smt2 <file-name>}\\
 	\item For CVC4 default with full-saturate-quant:\\ 
 	 \texttt{cvc4 --lang=smtlib2.6 --full-saturate-quant <file-name>}\\
 	\item  For Vampire default:\\ 
 	 \texttt{vampire -input\underline{ }syntax smtlib2 <file-name>} except the cases with superscripts. \begin{itemize}
 		\item  Superscript $*$:\\
 		 \texttt{vampire --input\underline{ }syntax smtlib2  <file-name>} \\
 		  \texttt{ --forced\underline{ }options "aac=none:add=large:afp=40000:afq=} \\ 
 		  \texttt{1.2:amm=off:anc=none: bd=off:fsr=off:gsp=input\underline{ }only} \\ \texttt{:inw=on:irw=on:lma=on:nm=64:nwc=1:sos=on:sp=}\\
 		   \texttt{occurrence:tha=off:} \texttt{updr=off:awr=5:s=1011:sa} \\ \texttt{=discount:ind=math"}\\
 		
 		\item  Superscript $\dagger$:\\ \texttt{vampire -input\underline{ }syntax smtlib2 -thsq on -thsqd 6}\linebreak \texttt{-thsqc 6 -thsqr 10,1  <file-name> } to prove the property and its lemmas, if any.\\
 		
 		\item  Superscript $\ddagger$:\\ options from $*$ to prove the inductive property in \texttt{ind\underline{ }property\underline{ }id.smt} and options from $\dagger$ to prove the actual property and its lemmas, if any.
 	\end{itemize}
 \end{itemize}
 
\noindent {\bf Encoding of \texttt{tranferN}, \texttt{uf} version:}
\begin{lstlisting}
(set-logic UFLIA)
(declare-sort Coin 0 )
(declare-sort Address 0)
(declare-fun act (Coin Int) Bool )
(declare-fun hc (Address Coin Int) Bool)
(declare-fun induct (Int) Bool)
(declare-const a1 Address)
(declare-const a2 Address)
(declare-const n Int)

;### inductive predicate definition ###
(assert (forall ((I Int)) 
 (= (induct I) 
    (and (forall ((C Coin))
          (= (exists ((A Address)) (hc A C I)) (act C I)) ))
         (forall ((A Address) (B Address) (C Coin))
          (=> (and (hc A C I) (hc B C I)) (= A B)) ))))) 
 
;### pre-invariants ###
;#inactive coins and at least one
(assert (forall ((C Coin))
(= (exists ((A Address)) (hc A C 0))
   (act C 0) )))
;#at most one
(assert (forall ((A Address)(B Address)(C Coin))
(=> (and (hc A C 0) (hc B C 0)) (= A B) )))

;### transition ###
(assert (forall ((I Int)) (=> 
 (<= 0 I)
 (and (forall ((D Coin)) (= (act D I) (act D (+ I 1)) ))
      (exists ((C Coin)) (and
       (hc a1 C I) (not (hc a2 C I))
       (not (hc a1 C (+ I 1))) (hc a2 C (+ I 1))
       (forall ((D Coin) (A Address))
        (=> (or (distinct C D)
                (and (distinct A a1) (distinct A a2)) )
            (= (hc A D (+ I 1)) (hc A D I)) ))))))))

;### negated post-invariant ###
(assert (and (<= 0 n) (not (induct n)) ))

(check-sat)
\end{lstlisting}

 \subsection{Codes for Table \ref{tbl:arith_table}}
 The precise function calls used are as follows.
\begin{itemize}
	\item For Z3 default:\\  \texttt{z3 -smt2 <file-name>}\\
	\item For CVC4 default with full-saturate-quant:\\  \texttt{cvc4 --lang=smtlib2.6 --full-saturate-quant <file-name>}\\
	\item  For Vampire default:\\  \texttt{vampire -input\underline{ }syntax smtlib2 <file-name>}\\
\end{itemize}

\noindent {\bf Encoding of  $\bal(a_0)+5$, $\bal(a_1)-3$, $\bal(a_2)-1$:}
\begin{lstlisting}
(set-logic UFLIA)
(declare-sort Coin 0 )
(declare-sort Address 0)
(declare-fun old-sum () Int)
(declare-fun new-sum () Int)
(declare-fun a0 () Address)
(declare-fun a1 () Address)
(declare-fun a2 () Address)
(declare-fun old-bal (Address) Int)
(declare-fun new-bal (Address) Int)
(declare-fun count (Coin) Int)
(declare-fun ind (Coin Address) Int)

;###  axioms on sum and count  ###
;#sum non-negative
(assert (<= 0 old-sum))
(assert (<= 0 new-sum))
;#count positive
(assert (forall ((C Coin)) (< 0 (count C)) ) )
;#count injective
(assert (forall ((C Coin) (D Coin))
 (=> (= (count C) (count D)) (= C D) )))
;#count instances of surjectivity
(assert (exists ((C Coin)) (= (count C) old-sum) ))
(assert (exists ((C Coin)) (= (count C) (+ old-sum 1)) ))
(assert (exists ((C Coin)) (= (count C) new-sum) ))

;####  axioms on bal and ind ###
;#bal non-negative
(assert (forall ((A Address)) (<= 0 (old-bal A)) ))
(assert (forall ((A Address)) (<= 0 (new-bal A)) ))
;#ind positive
(assert (forall ((C Coin)(A Address)) (< 0 (ind C A)) ))
;#ind(A,.) injective
(assert (forall ((C Coin) (D Coin) (A Address))
 (=> (= (ind C A) (ind D A)) (= C D) )))
;#ind(A,.) instances of surjectivity
(assert (exists ((C Coin)) (= (ind C a0) (old-bal a0)) ))
(assert (exists ((C Coin)) (= (ind C a0) (+(old-bal a0) 1))))
(assert (exists ((C Coin)) (= (ind C a0) (+(old-bal a0) 2))))
(assert (exists ((C Coin)) (= (ind C a0) (+(old-bal a0) 3))))
(assert (exists ((C Coin)) (= (ind C a0) (+(old-bal a0) 4))))
(assert (exists ((C Coin)) (= (ind C a0) (+(old-bal a0) 5))))
(assert (exists ((C Coin)) (= (ind C a1) (new-bal a1)) ))
(assert (exists ((C Coin)) (= (ind C a1) (+(new-bal a1) 1))))
(assert (exists ((C Coin)) (= (ind C a1) (+(new-bal a1) 2))))
(assert (exists ((C Coin)) (= (ind C a1) (+(new-bal a1) 3))))
(assert (exists ((C Coin)) (= (ind C a2) (new-bal a2)) ))
(assert (exists ((C Coin)) (= (ind C a2) (+(new-bal a2) 1))))

;###  axioms between sum and bal  ###
; #ind leq bal iff count leq sum
(assert (forall ((C Coin)) (=
 (exists ((A Address)) (<= (ind C A) (old-bal A)) )
 (<= (count C) old-sum) )))
(assert (forall ((C Coin)) (=
 (exists ((A Address)) (<= (ind C A) (new-bal A)) )
 (<= (count C) new-sum) )))
;#only once ind leq bal
(assert (forall ((A Address)(B Address)(C Coin))
 (=> (and
       (<= (ind C A) (old-bal A) )
       (<= (ind C B) (old-bal B) ) )
     (= A B) )))
(assert (forall ((A Address)(B Address)(C Coin))
 (=> (and
       (<= (ind C A) (new-bal A) )
       (<= (ind C B) (new-bal B) ) )
     (= A B) )))

;###  transition and negated impact  ###
;#plus5 minus3 minus1
(assert (and 
   (= (new-bal a0) (+ (old-bal a0) 5))
   (= (old-bal a1) (+ (new-bal a1) 3))
   (= (old-bal a2) (+ (new-bal a2) 1))
   (forall ((A Address)) (=>
     (and (distinct A a0) (distinct A a1) (distinct A a2))
     (= (old-bal A) (new-bal A)) ))))

;#negated Impact
(assert (distinct (+ old-sum 1) new-sum) )

(check-sat)
\end{lstlisting}

\end{document}